\documentclass[lettersize,journal]{IEEEtran}
\usepackage{cite}
\usepackage{amsmath,amssymb,amsfonts}
\usepackage{algorithmic}
\usepackage{algorithm}
\usepackage{array}
\usepackage[caption=false,font=normalsize,labelfont=sf,textfont=sf]{subfig}
\usepackage{textcomp}
\usepackage{stfloats}  
\usepackage{graphicx}
\usepackage{bm}
\usepackage{cases}
\usepackage{xcolor}
\usepackage{enumerate} 
\usepackage{lipsum} 
\usepackage{url}
\newtheorem{theorem}{\bf{\textit{Theorem}}}[section]
\newtheorem{lemma}{\bf{\textit{Lemma}}}[section]
\usepackage{verbatim}
\usepackage{makecell}
\usepackage{graphicx}
\usepackage{cite}
\usepackage{hyperref}
\begin{document}

\title{Low-Cost Physical-Layer Security Design for IRS-Assisted mMIMO Systems with One-Bit DACs}

\author{Weijie Xiong, Jian Yang, Jingran Lin, Hongli Liu, Zhiling Xiao, and Qiang Li
\thanks{This work was supported by the National Natural Science Foundation of China under Grant 62171110. \textit{(Corresponding author: Jingran Lin)}.}
\thanks{Weijie Xiong, Hongli Liu and Zhiling Xiao are with the School of Information and Communication Engineering, University of Electronic Science and Technology of China, Chengdu 611731, China (e-mail: 202311012313@std.uestc.edu.cn; hongliliu@std.uestc.edu.cn; zhilingxiao9928@163.com).}
\thanks{Jian Yang is with School of Cyberspace Science and Technology, Beijing Institute of Technology, Beijing 100081, China, and also with the Laboratory
of Electromagnetic Space Cognition and Intelligent Control, Beijing 100083, China (e-mail: jian\_yang\_jy@163.com).}
\thanks{Jingran Lin and Qiang Li are with the School of Information and Communication Engineering, University of Electronic Science and Technology of China, Chengdu 611731, China, and also with the Tianfu Jiangxi Laboratory, Chengdu, Sichuan 641419, China (e-mail: jingranlin@uestc.edu.cn; lq@uestc.edu.cn).}
}

\markboth{Journal of \LaTeX\ Class Files,~Vol.~14, No.~8, August~2021}%
{Shell \MakeLowercase{\textit{et al.}}: A Sample Article Using IEEEtran.cls for IEEE Journals}


\maketitle

\begin{abstract}
Integrating massive multiple-input multiple-output (mMIMO) systems with intelligent reflecting surfaces (IRS) presents a promising paradigm for enhancing physical-layer security (PLS) in wireless communications. However, deploying high-resolution quantizers in large-scale mMIMO arrays, along with numerous IRS elements, leads to substantial hardware complexity. To address these challenges, this paper proposes a cost-effective PLS design for IRS-assisted mMIMO systems by employing one-bit digital-to-analog converters (DACs). The focus is on jointly optimizing one-bit quantized precoding at the transmitter and constant-modulus phase shifts at the IRS to maximize the secrecy rate. This leads to a highly non-convex fractional secrecy rate maximization (SRM) problem. To efficiently solve this problem, two algorithms are proposed: (1) the WMMSE-PDD algorithm, which reformulates the SRM problem into a sequence of non-fractional programs with auxiliary variables using the weighted minimum mean-square error (WMMSE) method and solves them via the penalty dual decomposition (PDD) approach, achieving superior secrecy performance; and (2) the exact penalty product Riemannian gradient descent (EPPRGD) algorithm, which transforms the SRM problem into an unconstrained optimization over a product Riemannian manifold, eliminating auxiliary variables and enabling faster convergence with a slight trade-off in secrecy performance. Both algorithms provide analytical solutions at each iteration and are proven to converge to Karush–Kuhn–Tucker (KKT) points. Simulation results confirm the effectiveness of the proposed methods and highlight their respective advantages.
\end{abstract}


\section{Introduction}
\IEEEPARstart{T}{he} broadcast nature of wireless transmissions makes them vulnerable to unauthorized interception, leading to eavesdropping attacks \cite{jiang2024physical}. To address this challenge, physical-layer security (PLS) has emerged as a fundamental paradigm, exploiting the statistical spatial characteristics of wireless channels to secure communications. By steering signals toward legitimate users while suppressing eavesdroppers, PLS enables reliable decoding and mitigates information leakage \cite{liu2024survey}. Owing to its ability to enhance security independently of cryptographic techniques, PLS has attracted increasing attention as a key research direction in wireless communications \cite{xiong2025secure}.

Meanwhile, massive multiple-input multiple-output (mMIMO) is regarded as one of the key technologies for future wireless communications. By deploying large-scale antennas at the transmitter, it not only provides substantial spectral gains but also opens new opportunities for enhancing PLS \cite{xiong2025securehb,zhu2014secure,8714023,wu2016secure,kapetanovic2015physical}. For instance, \cite{zhu2014secure} studied multicell mMIMO systems employing maximum ratio transmission (MRT) and null-space artificial noise (AN) to improve secrecy under both perfect and imperfect channel state information (CSI). In \cite{8714023}, matrix polynomial-based precoding was proposed to further strengthen secure transmission. Additionally, \cite{wu2016secure} investigated active eavesdropping and derived corresponding secrecy rates under matched filter precoding. Studies such as \cite{kapetanovic2015physical} developed detection strategies for active eavesdroppers by leveraging the abundant spatial DoFs in mMIMO. Nevertheless, secrecy performance can still degrade when the eavesdropping channel is strongly correlated with the legitimate channel.

Recent advances have identified intelligent reflecting surfaces (IRSs) as a promising solution to address this limitation and enhance PLS \cite{xiong2025enhancing}. Typically, an IRS comprises a metasurface of numerous passive elements capable of digitally adjusting the phase of incident electromagnetic waves, thereby reconfiguring the wireless environment and introducing additional reflection paths \cite{cheng2023ris}. Integrating IRSs with mMIMO mitigates channel correlation between legitimate users and eavesdroppers by reinforcing desired signals through reflected paths \cite{yang2024spatially}. Recent studies have demonstrated the effectiveness of IRS-assisted PLS in mMIMO networks \cite{asaad2022secure,kong2024distributed,xu2021artificial,vu2025enhancing,wang2022intelligent,hua2024unification}. For instance, \cite{asaad2022secure} jointly optimized the precoding at the mMIMO BS and the IRS phase shifts to maximize secrecy rates, achieving a significant improvement in secure performance compared to the scenario without IRS. To further strengthen PLS, \cite{kong2024distributed} and \cite{xu2021artificial} developed AN-assisted transmission strategies against passive eavesdropping. In more practical settings, \cite{vu2025enhancing} and \cite{wang2022intelligent} tackled the challenges posed by non-colluding and colluding eavesdroppers by jointly optimizing the transmit precoders and IRS configurations. Moreover, \cite{hua2024unification} explored the design with imperfect CSI and assessed the effectiveness of IRS in enhancing secrecy when the eavesdropping channel outperforms the legitimate channel.

However, existing research on PLS in IRS-assisted mMIMO systems predominantly relies on high-resolution quantizers, requiring numerous radio frequency (RF) chains and high-performance data converters \cite{ngo2024ultradense}. This significantly increases hardware complexity, particularly in large-scale deployments. Moreover, IRSs typically comprise hundreds of passive reflecting elements fabricated from specialized materials, such as metamaterials or phase-change materials \cite{sarkar2024comprehensive}. Although each element operates with low power, large-scale integration substantially raises overall energy consumption and manufacturing cost. These combined burdens critically limit the scalability and practical deployment of IRS-assisted mMIMO systems.

To address these challenges, a promising solution is to employ precoding architectures based on low-resolution components, such as one-bit digital-to-analog converters (DACs) \cite{cheng2021transmit}. As detailed in Table \ref{tab:power}, this architecture greatly reduces power consumption and enhances energy efficiency. Specifically, replacing 12-bit high-resolution DACs with one-bit variants reduces DAC power consumption by over 90\% (from $\approx 200$ mW to $\approx 5$ mW) \cite{lee2008analog}. Furthermore, the relaxed linearity requirements enable the deployment of Class-E power amplifiers (PAs) with $\sim 70\%$ efficiency, significantly outperforming the $\approx 30\%$ efficiency of conventional Class-AB PAs \cite{saxena2017analysis}. Recent studies have explored the integration of one-bit DACs into IRS-assisted mMIMO systems to improve energy efficiency and transmission performance \cite{wei2021waveform,arfaoui2023analysis,wu2019intelligent,chen2024joint,ci2025hybrid,li2024reconfigurable}. For instance, \cite{wei2021waveform} jointly optimized one-bit DAC precoding and IRS phase shifts to maximize system throughput under quantization constraints. In \cite{arfaoui2023analysis}, a deep learning (DL)-based precoding method was proposed to mitigate distortion caused by coarse quantization in IRS-aided networks. Furthermore, \cite{wu2019intelligent} and \cite{chen2024joint} investigated robust precoding algorithms to enhance spectral efficiency in one-bit DAC systems. Studies such as \cite{ci2025hybrid} focused on hardware-efficient joint optimization of hybrid precoding and IRS configurations, while \cite{li2024reconfigurable} extended the framework to multi-IRS architectures to further reduce pilot overhead and complexity.

\begin{table}[!t]
\centering
\caption{{Estimated Hardware Power Consumption per RF Chain}}
\label{tab:power}
\renewcommand{\arraystretch}{1.2}
\begin{tabular}{|l|c|c|c|}
\hline
\textbf{Component} & \textbf{12-bit} & \textbf{one-bit} & \textbf{Reduction} \\ \hline
DAC Power & $\approx 200$ mW & $\sim 5$ mW & $\approx 97.5\%$ \\ \hline
PA Efficiency & \makecell{Class-AB \\ ($\approx 30\%$)} & \makecell{Class-E \\ ($\approx 70\%$)} & High Efficiency \\ \hline
\textbf{Total RF Cost} & \textbf{High} & \textbf{Low} & \textbf{Significant} \\ \hline
\end{tabular}
\end{table}

Although one-bit DACs have been widely adopted in IRS-assisted mMIMO systems to reduce hardware costs, their effectiveness in achieving satisfactory low-cost PLS remains largely unexplored. Motivated by the hardware efficiency of one-bit DACs, this paper investigates their integration to enhance PLS performance in IRS-assisted mMIMO networks. Specifically, we consider a scenario in which an mMIMO transmitter communicates with a multi-antenna legitimate receiver in the presence of a multi-antenna eavesdropper. The objective is to jointly optimize one-bit quantized precoding at the transmitter and constant-modulus phase shifts at the IRS to maximize the secrecy rate.

However, addressing PLS in one-bit DAC–based IRS-assisted systems is substantially more challenging than conventional energy efficiency maximization \cite{wei2021waveform,arfaoui2023analysis,wu2019intelligent,chen2024joint,ci2025hybrid,li2024reconfigurable}. This difficulty mainly stems from the difference-based secrecy objective, which leads to a highly intractable mixed-integer fractional programming problem. Consequently, relaxation-based methods such as semidefinite relaxation (SDR) \cite{wei2021waveform,chen2024joint,li2024reconfigurable} and direct projection \cite{wu2019intelligent,ci2025hybrid} often suffer severe performance degradation, since the approximation error is significantly amplified by the difference-based secrecy objective. Moreover, existing DL architectures \cite{arfaoui2023analysis} cannot be directly applied, as the fractional nature of the secrecy objective induces adversarial gradients, rendering standard training algorithms highly unstable.

To address these challenges, we propose an optimization framework that fundamentally resolves the conflict between discrete constraints and fractional objectives. We first reformulate the one-bit constraints into continuous equality constraints via an exact transformation, introducing zero approximation error. Based on this formulation, we develop the WMMSE-PDD and EPPRGD algorithms using penalty dual decomposition and Riemannian manifold optimization, respectively. Crucially, both algorithms guarantee theoretical convergence to a KKT point of the original problem, thereby achieving superior secrecy performance compared to traditional relaxation-based methods. The main contributions are summarized as follows:
\begin{itemize}
\item Unlike previous studies on PLS in IRS-assisted mMIMO systems that rely on high-resolution quantizers at BS, this work focuses on a transmitter architecture employing low-resolution one-bit DACs. By significantly reducing the number of RF chains and the requirement for high-performance data converters, this approach lowers both energy consumption and hardware costs.
\item To tackle the non-convex fractional secrecy rate maximization (SRM) problem, we first propose a performance-focused algorithm that integrates the weighted minimum mean-square error (WMMSE) method with penalty dual decomposition (PDD). The WMMSE method reformulates the fractional objective into a sequence of tractable non-fractional problems with auxiliary variables, which are iteratively solved via PDD with guaranteed convergence to a Karush–Kuhn–Tucker (KKT) point.
\item We further propose a time-efficient exact penalty product Riemannian gradient descent (EPPRGD) algorithm that accelerates convergence without introducing auxiliary variables, compared to the WMMSE-PDD algorithm, at the cost of only a slight performance loss. Specifically, the exact penalty (EP) method is employed to reformulate the original problem into an unconstrained optimization over a product Riemannian manifold, which is then solved using the product Riemannian gradient descent (PRGD) algorithm, also guaranteeing convergence to a KKT point.
\end{itemize}

The structure of the paper is organized as follows. Section II introduces the system model and problem formulation. Section III presents the proposed performance-centered WMMSE-PDD algorithm, while Section IV further proposes the EPPRGD algorithm with faster convergence. Section VI presents the simulation results and Section VII concludes the paper.

The following notations are used throughout the paper. Vectors and matrices are denoted by $\mathbf{a}$ and $\mathbf{A}$, respectively. The operators $(\cdot)^T$, $(\cdot)^H$, and $(\cdot)^*$ denote the transpose, conjugate transpose, and conjugate of a matrix or vector, respectively. $\mathbf{I}$ denotes the identity matrix, and $\mathbb{C}^N$ represents the space of $N$-dimensional complex vectors. The complex Gaussian distribution with mean $\mu$ and variance $\sigma^2$ is denoted by $\mathcal{CN}(\mu, \sigma^2)$. The operators $\text{Tr}(\cdot)$, $\|\cdot\|_F$, $\|\cdot\|_2$, and $|\cdot|$ denote the trace, Frobenius norm, Euclidean norm, and absolute value, respectively. The phase, real part, and imaginary part of $\mathbf{A}$ are denoted by $\text{arg}(\mathbf{A})$, $\Re(\mathbf{A})$, and $\Im(\mathbf{A})$, respectively. The Hadamard division and product are denoted by $\mathbf{A} \oslash \mathbf{B}$ and $\mathbf{A} \odot \mathbf{B}$, respectively. The direct sum operation satisfies $\mathbf{A} \oplus \mathbf{B} = \text{blkdiag}(\mathbf{A}, \mathbf{B})$, where $\text{blkdiag}(\mathbf{A}, \mathbf{B})$ denotes a block diagonal matrix with $\mathbf{A}$ and $\mathbf{B}$ placed along the diagonal and zeros elsewhere.

\section{System Model and Problem Formulation}

Consider an mMIMO downlink network assisted by an IRS, where the IRS is deployed to assist legitimate communication between the transmitter (Alice) and the legitimate user (Bob) in the presence of an eavesdropper (Eve), which attempts to intercept information from the communication. The number of antennas at Alice, Bob, and Eve are \(M\), \(N_b\), and \(N_e\), respectively, while the IRS consists of \(N_i\) reflecting coefficients. 

In order to reduce hardware cost, Alice is equipped with one-bit DACs. Specifically, let $\mathbf{s}\in\mathbb{C}^{M}$ denote the transmit signal. Each entry of \(\mathbf{s}\) is passed through two one-bit quantizers that independently quantize its real and imaginary parts. The resulting quantized transmit vector is then given by,  
\begin{equation}
\mathbf{x} = P\mathcal{Q}_C(\mathbf{s}) 
\in\mathcal{X}^{M},
\label{transmitsignal}
\end{equation}  
where $P$ is the total transmit power; \(\mathcal{Q}_C(\cdot)=\mathcal{Q}(\cdot)+j\mathcal{Q}(\cdot)\) represents the complex-valued element-wise quantization function, which consists of two one-bit real-valued quantizer $\mathcal{Q}(\cdot)$; $\mathcal{X}$ denote the transmit alphabet, given by,
\begin{equation}
\mathcal{X}=\left\{\pm \sqrt{\frac{1}{2M}}\pm j\sqrt{\frac{1}{2M}}\right\},
\end{equation} 
which corresponds to the QPSK signal. It should be noted that the transmit one-bit signal has constant-modulus, and hence, low-cost and linear power amplifiers can be adopted at the RF chains \cite{cheng2021transmit}.

Denote \( { \bf \hat H}_{ab} \in \mathbb{C}^{N_b \times M} \), \( { \bf \hat H}_{ae} \in \mathbb{C}^{N_e \times M} \), \( \mathbf{H}_{ai} \in \mathbb{C}^{N_i \times M} \), \( { \bf \hat H}_{ib} \in \mathbb{C}^{N_b \times N_i} \), and \( { \bf \hat H}_{ie} \in \mathbb{C}^{N_e \times N_i} \) as the channels between Alice-Bob, Alice-Eve, Alice-IRS, IRS-Bob, and IRS-Eve, respectively\footnote{In this paper, we assume that CSI of Eve is available to Alice. This can be achieved through various methods, such as the CSI feedback method or even the local oscillator power leakage from Eve's RF frontend \cite{mukherjee2012detecting}. Additionally, we assume that the transmitter has perfect CSI of Eve, which is feasible, for instance, if Eve is also a user of the system and the transmitter aims to provide different types of users with distinct services or content. These services or content should be provided exclusively to the target users.}. Let ${\bf{\Theta }} = diag({\bm \theta})$ denote the diagonal phase shift matrix for the IRS where ${\bm \theta} = [\theta_1,\theta_2,...,\theta_{N_i}]^T \in \mathbb{C}^{N_i}$. Assuming quasi-static channels, the received signals at Bob and Eve are expressed as,
\begin{equation}
\begin{aligned}
& \mathbf{y}_b = P( { \bf \hat H}_{ib} {\bf{\Theta}} \mathbf{H}_{ai} + { \bf \hat H}_{ab} ) \mathbf{x}+ {\bf{n}}_b \in \mathbb{C}^{N_b}, \\
& \mathbf{y}_e = P( { \bf \hat H}_{ie} {\bf{\Theta}} \mathbf{H}_{ai} + { \bf \hat H}_{ae}  )\mathbf{x} + {\bf{n}}_e \in \mathbb{C}^{N_e},\label{recievesignal}
\end{aligned}
\end{equation}
where \( {\bf{\Theta}} \in \mathbb{C}^{N_i \times N_i} \) denotes the reflection coefficient matrix for the IRS; \( {\bf{n}}_b \sim \mathcal{CN}(\mathbf{0}, \sigma_b^2\mathbf{I}_b) \) and \( {\bf{n}}_e \sim \mathcal{CN}(\mathbf{0}, \sigma_e^2\mathbf{I}_e) \) represent the vectors of zero-mean
additive white Gaussian noise (AWGN).

Based on (\ref{recievesignal}), the secrecy rate from Alice to Bob is given as \cite{asaad2022secure},
\begin{equation}
{R_{sec}({\bf x},{\bm \theta})}=\left[ {R_b({\bf x},{\bm \theta})}- {R_e({\bf x},{\bm \theta})} \right]^{+}, \label{UsetogCs}
\end{equation}
where \(R_b({\bf x},{\bm \theta}) \) and \( R_e({\bf x},{\bm \theta}) \) represent the achievable rates at Bob and Eve, respectively, and are formulated as \cite{asaad2022secure},
\begin{subequations}
\begin{align}
& {R_b({\bf x},{\bm \theta})}= \log \left({1}+||({\mathbf H}_{ib} {\bf{\Theta }} {\mathbf H}_{ai}+{\mathbf H}_{ab} ) {\bf x}||_2^2\right), \\
&  {R_e({\bf x},{\bm \theta})}= \log \left({1}+||({\mathbf H}_{ie} {\bf{\Theta }} {\mathbf H}_{ai}+{\mathbf H}_{ae} ) {\bf x}||_2^2\right),
\end{align}
\label{trasnmittr}
\end{subequations}
where $ { \bf  H}_{ab} = \sqrt{{P}/{\sigma_b^2}}{ \bf \hat H}_{ae}$, $ { \bf  H}_{ae} = \sqrt{{P}/{\sigma_e^2}}{ \bf \hat H}_{ab}$, $ { \bf  H}_{ib} = \sqrt{{P}/{\sigma_b^2}}{ \bf \hat H}_{ib}$ and $ { \bf  H}_{ie} = \sqrt{{P}/{\sigma_e^2}}{ \bf \hat H}_{ie}$.

In this paper, we aim to maximize the secrecy rate by jointly optimizing the quantized precoding vector \(\mathbf{x}\) at Alice and the phase shifts vector \(\bm{\theta}\) at the IRS, subject to constraints on the one-bit precoding alphabet and the constant-modulus IRS phase shifts. Therefore, we formulate the SRM problem as follows,
\begin{subequations}
\begin{align}
 \max _{{\bf x}, {\bm \theta}} \quad& R_{sec}({\bf x},{\bm \theta}), \label{objorig}\\
 \text { s.t. }\quad &{\bf x}\in\mathcal{X}^{M}, \label{cons1orig}\\
 & |{\theta}_{n}| = 1, \forall n \in N_i. \label{cons2orig}
\end{align}
\label{overallproblem}%
\end{subequations}
Substituting the rate expressions from (\ref{UsetogCs}) and (\ref{trasnmittr}) into (\ref{objorig}), the objective function is explicitly expressed as,
\begin{equation}\label{uselognew}
    {R_{sec} = \max\left\{0, \log\left(\frac{1+\|({\mathbf H}_{ib} {\bf{\Theta }} {\mathbf H}_{ai}+{\mathbf H}_{ab} ) {\bf x}\|_2^2}{1+\|({\mathbf H}_{ie} {\bf{\Theta }} {\mathbf H}_{ai}+{\mathbf H}_{ae} ) {\bf x}||_2^2}\right)\right\}.}
\end{equation}
Observing that the objective function (\ref{uselognew}) yields the maximum between the constant 0 and the logarithmic term $\log(\frac{1+\|({\mathbf H}_{ib} {\bf{\Theta }} {\mathbf H}_{ai}+{\mathbf H}_{ab} ) {\bf x}\|_2^2}{1+\|({\mathbf H}_{ie} {\bf{\Theta }} {\mathbf H}_{ai}+{\mathbf H}_{ae} ) {\bf x}||_2^2})$, and since the logarithmic function is strictly monotonically increasing, the optimization can initially focus on maximizing this logarithmic term. Note that if the maximized logarithmic value is non-positive, the achievable secrecy rate is zero, rendering secure transmission unfeasible. Therefore, the problem can be reformulated as,
\begin{subequations}
\begin{align}
 \max _{{\bf x}, {\bm \theta}} \quad& \log\frac{{1}+||({\mathbf H}_{ib} {\bf{\Theta }} {\mathbf H}_{ai}+{\mathbf H}_{ab} ) {\bf x}||_2^2}{{1}+||({\mathbf H}_{ie} {\bf{\Theta }} {\mathbf H}_{ai}+{\mathbf H}_{ae} ) {\bf x}||_2^2}, \label{objorig2}\\
 \text { s.t. }\quad &{\bf x}\in\mathcal{X}^{M}, \label{cons1orig2}\\
 & |{\theta}_{n}| = 1, \forall n \in N_i. \label{cons2orig2}
\end{align}
\label{overallproblem2}%
\end{subequations}
Problem (\ref{overallproblem2}) constitutes a non-convex mixed-integer fractional programming problem, which is generally known to be NP-hard \cite{8714023}. The intractability primarily arises from four factors: (1) the non-concavity of the fractional objective function (\ref{objorig2}); (2) the coupling between variables ${\bf x}$ and ${\bm \theta}$; (3) the discrete nature of the one-bit constraint (\ref{cons1orig2}); and (4) the non-convex constant-modulus constraint (\ref{cons2orig2}). To address this, the following sections develop practical methods for solving problem (\ref{overallproblem2}).

\section{Proposed WMMSE-PDD Algorithm}
This section addresses the non-convex secrecy rate maximization problem (\ref{overallproblem}) by proposing an algorithm that combines WMMSE‐based reformulation with the PDD method. Firstly, we equivalently rewrite the one-bit quantized constraints (\ref{cons1orig2}) into a continuous smooth form. Subsequently, the fractional objective function (\ref{objorig}) is recast into a more tractable form using the WMMSE method. Finally, the PDD method is applied to solve the transformed problem. The details are provided below.

\subsection{Smoothing of the One-Bit Constraint (\ref{cons1orig2})}
To facilitate the solution of the optimization problem, we rewrite the quantized constraint in (\ref{cons1orig2}) into the equivalent smooth form given by Lemma \ref{lemma31}.

\begin{lemma} \label{lemma31}\itshape
Let \(\mathcal{X} =\bigl\{\pm\sqrt{\tfrac{1}{2M}}\pm j\sqrt{\tfrac{1}{2M}}\bigr\}\), and let \(\mathbf x\in\mathbb C^{M}\). Then, the following statements are equivalent,
\begin{equation}
x_m\in{\mathcal X}, \forall m
\Longleftrightarrow
\begin{cases}
{\mathbf x}^H{\mathbf x} = 1,\\
-\sqrt{\tfrac{1}{2M}}\le\Re\{x_m\}\le\sqrt{\tfrac{1}{2M}},\forall m,\\
-\sqrt{\tfrac{1}{2M}}\le\Im\{x_m\}\le\sqrt{\tfrac{1}{2M}},\forall m.
\end{cases}\label{refor1}
\end{equation}
\end{lemma}

{\bf\textit{Proof}}: See Appendix \ref{MPDD1}. $\hfill\blacksquare$

\subsection{Problem Transformation by WMMSE method}
The objective in (\ref{objorig2}) is a ratio of quadratic functions in $\mathbf{x}$ and $\bm{\theta}$, resulting in a challenging fractional programming problem. To ease the difficulty of directly handling the fractional objective, we employ the WMMSE method to reformulate it into an equivalent non-fractional form. Specifically, the WMMSE algorithm converts the SRM into an equivalent, more tractable minimization problem by introducing auxiliary variables. The key properties underlying this equivalence are summarized in Lemma \ref{lemma32}.
\begin{lemma} \label{lemma32} \itshape
(Lemma 4.1 in \cite{shi2015secure}) The following facts hold true.

(1) For any scalar $e\in\mathbb{C}$ with $\Re\{e\}\ge0$, we have,
\begin{equation}
  -\log(e)=\max_{w>0}\log(w)-we+b.
\end{equation}

(2) Define a function,
\begin{equation}
    {\mathbb{E}({\bf v},{\bf x})=(1-{\bf v}^H{\bf H}{\bf x})(1-{\bf v}^H{\bf H}{\bf x})^H+n{\bf v}^H{\bf v},}
\end{equation}
where $n$ is a positive scalar, we have,
\begin{equation}
\begin{aligned}
   \log(1+\|{\bf H}{\bf x}\|_2^2/n)=\max_{w>0,{\bf v}}\log(w)-w\mathbb{E}({\bf v},{\bf x})+b.
\end{aligned}    
\end{equation}

\end{lemma}

Next, using Lemma \ref{lemma32}, we derive an equivalent problem of problem (\ref{overallproblem2}) by introducing some auxiliary variables. Define,
\begin{equation}
    \mathbb{E}({\bf v},{\bf x})=(1-{\bf v}^H{\bf H}_b({\bm \theta}){\bf x})(1-{\bf v}^H{\bf H}_b({\bm \theta}){\bf x})^H+{\bf v}^H{\bf v},
\end{equation}
where ${\bf H}_b({\bm \theta}) = {\mathbf H}_{ib} {\bf{\Theta }} {\mathbf H}_{ai}+{\mathbf H}_{ab} $. Then we have from fact (2) that,
\begin{equation}
    \log\left({1}+||{\bf H}_b({\bm \theta}) {\bf x}||_2^2\right)=\max_{w_b>0,{\bf v}} \log(w_b)-w_b\mathbb{E}({\bf v},{\bf x})+b_b,\label{refor2}
\end{equation}
Furthermore, from fact (1), we have,
\begin{equation}
\begin{aligned}
  &-  \log\left({1}+||{\bf H}_e({\bm \theta}) {\bf x}||_2^2\right)\\&=\max_{w_e>0} \log(w_e)-w_e({1}+||{\bf H}_e({\bm \theta}) {\bf x}||_2^2)+b_e,\label{refor3}
\end{aligned}  
\end{equation}
where ${\bf H}_e({\bm \theta}) = {\mathbf H}_{ie} {\bf{\Theta }} {\mathbf H}_{ai}+{\mathbf H}_{ae}$. Based on (\ref{refor1}), (\ref{refor2}), and (\ref{refor3}), after ignoring the constant terms, problem (\ref{overallproblem2}) is equivalent to,
\begin{subequations}
\begin{align}
   \min_{w_b,w_e,{\bf v},{\bf x},{\bm \theta}} &\left\{\begin{array}{l}w_b{\bf x}^H{\bf H}_b^H({\bm \theta}){\bf v}{\bf v}^H{\bf H}_b({\bm \theta}){\bf x}\\-w_b({\bf x}^H{\bf H}_b^H({\bm \theta}){\bf v}+{\bf v}^H{\bf H}_b({\bm \theta}){\bf x})\\+w_b({\bf v}^H{\bf v}+1)-\log(w_b)\\+w_e{\bf x}^H{\bf H}^H_e({\bm \theta}) {\bf H}_e({\bm \theta}){\bf x}+w_e-\log(w_e)
 \end{array}\right\},\label{wwse1}\\
 \text { s.t. }\quad  & -\sqrt{\tfrac{1}{2M}}\le\Re\{x_m\},\Im\{x_m\}\le\sqrt{\tfrac{1}{2M}},\forall m,\label{wwse3}\\
  & {\mathbf x}^H{\mathbf x} = 1,\label{wwse2}\\
 & |{\theta}_{n}| = 1, \forall n. \label{wwse4} 
\end{align}
\label{newforMMSE1}%
\end{subequations}
Although the reformulations improve tractability, the problem remains challenging due to the non-convex, highly coupled objective function (\ref{wwse1}), the coupled constraints (\ref{wwse2}) and (\ref{wwse3}), and the non-convex constraints (\ref{wwse4}). We observe that when either variable is fixed, the resulting subproblem reduces to a convex objective, which is considerably easier to solve. Furthermore, by introducing auxiliary variables and incorporating them into the objective through penalty terms, the otherwise difficult constraints become more manageable. Motivated by these observations, we exploit the sub-convex structure of the objective and the penalty-based formulation to develop an alternating minimization scheme based on the PDD method.

\subsection{Realization of the PDD-based Algorithm}
The PDD method is a double-loop framework \cite{shi2020penalty}, in which the inner loop utilizes the block coordinate descent (BCD) approach to iteratively solve the augmented Lagrangian problem, while the outer loop updates the dual variables or the penalty parameters. The details are described next.

To reduce the difficulty of handling the coupled non-convex constraints, auxiliary variables \({\mathbf{t}} \in \mathbb{C}^M\) and \({\bm{\phi}} \in \mathbb{C}^{N_i}\) are introduced to reformulate the optimization problem in (\ref{newforMMSE1}) as follows,
\begin{subequations}
\begin{align}
   &   \min_{\begin{array}{c} w_b,w_e,{\bf v},\\{\bf x},{\bm \theta},{\bf t},{\bm \phi}
\end{array}} \left\{\begin{array}{l}w_b{\bf x}^H{\bf H}_b^H({\bm \theta}){\bf v}{\bf v}^H{\bf H}_b({\bm \theta}){\bf x}\\-w_b({\bf x}^H{\bf H}_b^H({\bm \theta}){\bf v}+{\bf v}^H{\bf H}_b({\bm \theta}){\bf x})\\+w_b({\bf v}^H{\bf v}+1)-\log(w_b)\\+w_e{\bf x}^H{\bf H}^H_e({\bm \theta}) {\bf H}_e({\bm \theta}){\bf x}+w_e-\log(w_e)
 \end{array}\right\},\label{pdd1}\\
 &\quad\quad\text { s.t. }\quad -\sqrt{\tfrac{1}{2M}}\le\Re\{t_m\},\Im\{t_m\}\le\sqrt{\tfrac{1}{2M}},\label{pdd3}\\
  & \quad\quad\quad\quad\quad{\mathbf x}^H{\mathbf x} = 1,\label{pdd2}\\
 &\quad\quad\quad\quad\quad |{\phi}_{n}| = 1, \forall n, \label{pdd4} \\
  &\quad\quad\quad\quad\quad {\bf t} = {\bf x},  \label{pdd5} \\
  &\quad\quad\quad\quad\quad {\bm \phi} = {\bm \theta}.  \label{pdd6}
\end{align}
\label{pdd1}%
\end{subequations}
By penalizing the equality constraints (\ref{pdd5}) and (\ref{pdd6}) into the objective function, the augmented Lagrangian problem of (\ref{pdd1}) is written as,
\begin{equation}
\begin{aligned}
 \min _{w_b,w_e,{\bf v},{\bf x},{\bm \theta},{\bf t},{\bm \phi}} \quad &{\cal L}_{\rho} (w_b,w_e,{\bf v},{\bf x},{\bm \theta},{\bf t},{\bm \phi};{\bm \varphi }, {\bm \psi} ) \\&= f(w_b,w_e,{\bf v},{\bf x},{\bm \theta},{\bf t},{\bm \phi})\\ &+\frac{1}{2\rho}\|{\bf t}-{\bf x}+{\rho}{\bm \varphi }\|_2^2+\frac{1}{2\rho}\|{\bm \phi}-{\bm \theta}+{\rho}{\bm \psi }\|_2^2, \\
 \text {s.t.} \quad\quad\quad&(\ref{pdd3}), (\ref{pdd2}), (\ref{pdd4}) \text{ are satisfied}, 
\end{aligned}
\label{Lagrangian}
\end{equation}
where $ f(w_b,w_e,{\bf v},{\bf x},{\bm \theta},{\bf t},{\bm \phi})$ represents the objective function in (\ref{pdd1}); ${\bm \varphi} \in \mathbb{C}^{M}$ and ${\bm \psi }\in \mathbb{C}^{N_i}$ are the dual variables; $\rho > 0$ is the penalty parameter.  
Note that the variables \(\{  w_b,w_e,{\bf v},{\bf x},{\bm \theta},{\bf t},{\bm \phi} \}\) are fully separable in the constraints of (\ref{Lagrangian}). Specifically, when the other variables are fixed, each resulting subproblem with respect to a single variable becomes relatively simple. These observations motivate the use of the BCD strategy \cite{tseng2001convergence} to iteratively solve (\ref{Lagrangian}), forming the inner loop of the algorithm. Although some subproblems may involve non-convex constraints, we will later show that each admits a unique global optimum due to its specific structure, thereby ensuring that the BCD-based algorithm converges to a KKT point. The overall process of the PDD-based algorithm is shown in Fig. \ref{PDD}. In the following, we elaborate on the details of the inner-loop BCD-based algorithm.

\begin{figure}[t]
  \begin{center}
  \includegraphics[width=2.5in]{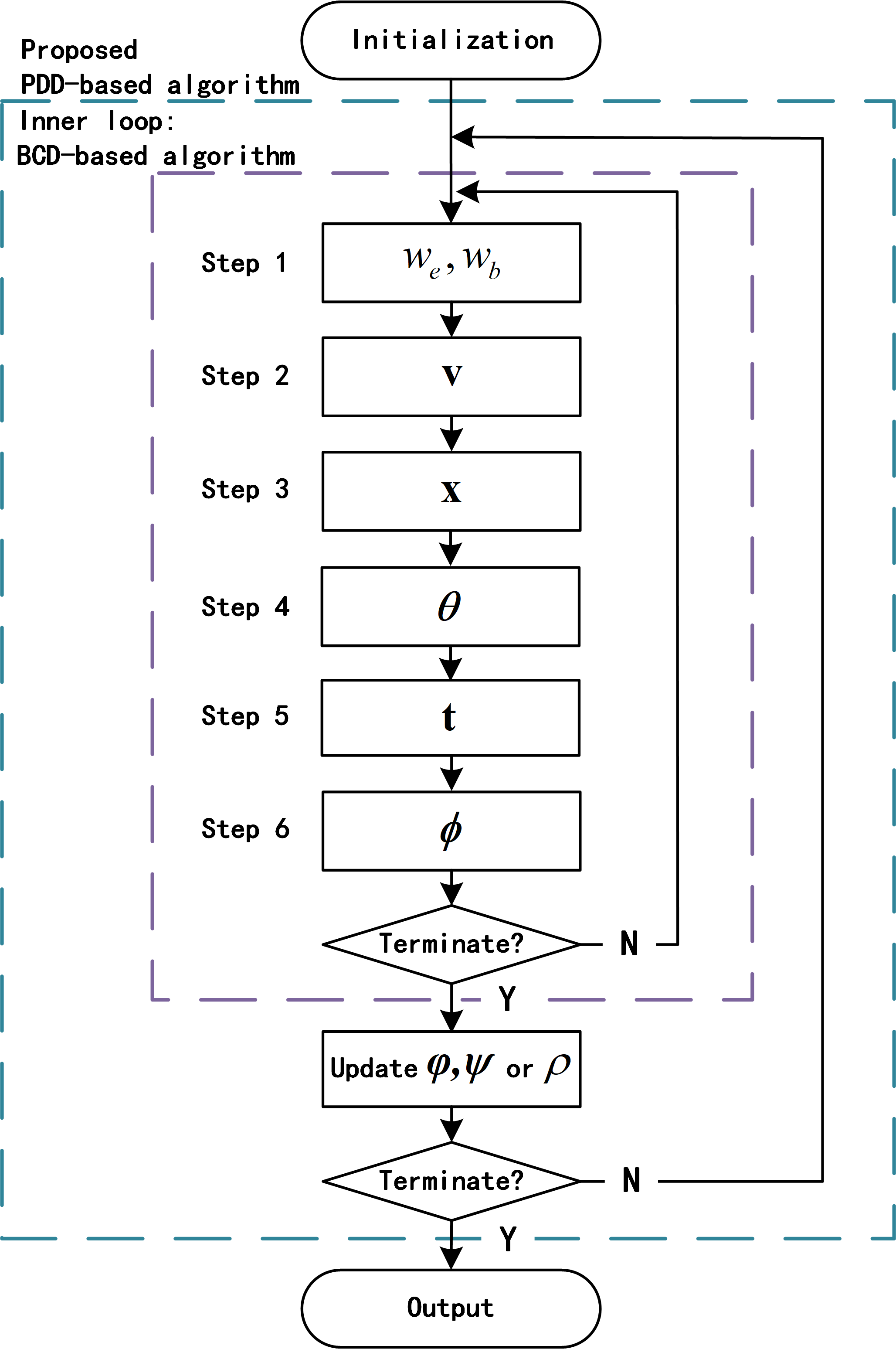}\\
  \caption{ The structure of the proposed PDD-based algorithm.}\label{PDD}
  \end{center}
\end{figure}

\underline{\textit{Inner Layer Procedure}}\\
\subsubsection{Update $w_b$ and $w_e$}
With the other variables being fixed, the objective function of problem (\ref{Lagrangian}) is separable over $w_b$ and $w_e$. Hence, Using fact (1) in Lemma \ref{lemma32} twice, we can easily obtain the optimal $w_b$ and $w_e$ given as,
\begin{equation}
\begin{aligned}
    &w_b = ((1-{\bf v}^H{\bf H}_b({\bm \theta}){\bf x})(1-{\bf v}^H{\bf H}_b({\bm \theta}){\bf x})^H+{\bf v}^H{\bf v})^{-1},\\
    &w_e = (1+\|{\bf H}_e({\bm \theta}){\bf x}\|_2^2)^{-1}.
\end{aligned}   
\label{upwewb}%
\end{equation}

\subsubsection{Update $\bf v$}
The subproblem with respect to \(\mathbf{v}\) is a convex unconstrained optimization problem. By examining the first-order optimality condition \cite{boyd2004convex}, the optimal \(\mathbf{v}\) is given as,
\begin{equation}
    \mathbf{v} = (1+ \|{\bf H}_b({\bm \theta}){\bf x}\|_2^2)^{-1}{\bf H}_b({\bm \theta}){\bf x}. \label{upv}
\end{equation}

\subsubsection{Update $\bf x$}
Given the other variables, the subproblem related to $\bf x$ is given by,
\begin{subequations}
\begin{align}
    \min_{{\bf x}} \quad& {\bf x}^H {\bf A} {\bf x} - {\bf x}^H{\bf b}-{\bf b}^H{\bf x}+\frac{1}{2\rho}\|{\bf x}-({\bf t}+{\rho}{\bm \varphi })\|_2^2,\\
  \text { s.t. }\quad &{\mathbf x}^H{\mathbf x} = 1, \label{upxq}
\end{align}  
\label{upxnow1}%
\end{subequations}
where ${\bf A}=w_b{\bf H}_b^H({\bm \theta}){\bf v}{\bf v}^H{\bf H}_b({\bm \theta})+w_e{\bf H}^H_e({\bm \theta}) {\bf H}_e({\bm \theta})$, ${\bf b}=w_b{\bf H}_b^H({\bm \theta}){\bf v}$. Due to the equality constraint, problem (\ref{upxnow1}) is non-convex. However, by examining the KKT conditions \cite{boyd2004convex}, we can derive a closed-form solution. The first-order optimality condition yields the following expression for ${\bf x}(\lambda)$,
\begin{equation}
   { {\bf x}(\lambda) = (2\rho{\bf A}+(1+2\rho\lambda){\bf I}_M)^{-1}({\bf t}+{\rho}{\bm \varphi }+2\rho{\bf b}),} \label{upx}
\end{equation}
where $\lambda$ is the Lagrange multiplier. To guarantee a global minimizer, the scaled Hessian ${\bf K} = 2\rho{\bf A}+(1+2\rho\lambda){\bf I}_M$ is required to be positive definite \cite{boyd2004convex}, which restricts $\lambda$ to $(-\tfrac{1}{2\rho},\infty)$ given ${\bf A} \succeq 0$. Moreover, define $g(\lambda) \triangleq \|{\bf x}(\lambda)\|_2^2$. By checking the spectral decomposition of ${\bf A}$ with eigenvalues $\{\mu_m\}_{m=1}^M$ \cite{boyd2004convex}, $g(\lambda)$ can be explicitly expressed as,
\begin{equation}
   { g(\lambda) = \sum_{m=1}^M \frac{|[{\bf U}^H({\bf t}+{\rho}{\bm \varphi }+2\rho{\bf b})]_m|^2}{(2\rho\mu_m + 1 + 2\rho\lambda)^2},}
\end{equation}
where ${\bf U}$ is the eigenvector matrix of ${\bf A}$. This formulation confirms that $g(\lambda)$ is strictly monotonically decreasing. Combined with the limits $\lim_{\lambda \to (-1/2\rho)^+} g(\lambda) = \infty$ and $\lim_{\lambda \to \infty} g(\lambda) = 0$, the existence of a unique $\lambda^*$ is guaranteed. Consequently, the bisection method \cite{boyd2004convex} can be employed to efficiently locate the unique $\lambda$. Substituting this $\lambda$ into equation (\ref{upx}) yields the optimal solution ${\bf x}={\bf x}(\lambda)$. Finally, Lemma \ref{lemma33} formally establishes that this solution is the unique global optimum.

\begin{lemma} \label{lemma33} \itshape
The solution given by (\ref{upx}) is the unique global optimum of problem (\ref{upxnow1}).
\end{lemma}

{\bf\textit{Proof}}: See Appendix \ref{Pro33}. $\hfill\blacksquare$

\subsubsection{Update $\bm \theta$}
Fixing the other variables, the update of $\bm \theta$ reduces to a convex unconstrained problem. Rearranging the relevant terms and omitting constants, it is given by,
\begin{equation}
    \min_{{\bm \theta }} \quad {\bm \theta }^H{\bf C}{\bm \theta }+2\Re\{{\bm \theta}^H{\bf d}^*\}+\frac{1}{2\rho}\|{\bm \phi}-{\bm \theta}+{\rho}{\bm \psi }\|_2^2, \label{uptheta}
\end{equation}
where, 
\begin{equation}
    \begin{aligned}
        &{\bf C}={\bf C}_b\odot{\bf E}^T+{\bf C}_e\odot{\bf E}^T,\quad {\bf E}={\bf H}_{ai}{\bf x}{\bf x}^H{\bf H}_{ai}^H, \\
        &{\bf C}_b=w_b{\bf H}_{ib}^H{\bf v}{\bf v}^H{\bf H}_{ib},\quad{\bf C}_e=w_e{\bf H}_{ie}^H{\bf H}_{ie},\\
        & {\bf d} = \text{diag}\left(\begin{array}{l}w_b{\bf H}_{ai}{\bf x}{\bf x}^H{\bf H}_{ab}^H{\bf v}{\bf v}^H{\bf H}_{ib}\\+w_e{\bf H}_{ai}{\bf x}{\bf x}^H{\bf H}_{ae}^H{\bf H}_{ie}-w_b{\bf H}_{ai}{\bf x}{\bf v}^H{\bf H}_{ib} \end{array}\right).
    \end{aligned}
\end{equation}
For more details of reformulation, readers are encouraged to consult \cite{xiong2025constant}. By checking the first-order optimality condition, the optimal $\bm \theta$ is given by,
\begin{equation}
   {\bm \theta } = (2{\rho}{\bf C}+{\bf I}_{N_i})^{-1}( {\bm \phi}+\rho{\bm \psi }-2{\rho}{\bf d}^* ). \label{resulttheta}
\end{equation}

\subsubsection{Update $\bf t$}
With other variables being fixed, the subproblem related to $\bf t$ is given as,
\begin{subequations}
\begin{align}
    \min_{\bf t} \quad& \frac{1}{2\rho}\|{\bf t}-{\bf x}+{\rho}{\bm \varphi }\|_2^2,\\
    \text { s.t. }\quad &-\sqrt{\tfrac{1}{2M}}\le\Re\{t_m\},\Im\{t_m\}\le\sqrt{\tfrac{1}{2M}}, \label{tcons}
\end{align}  
\label{updt}%
\end{subequations}
which is a convex optimization problem. Ignoring constraint (\ref{tcons}), the unconstrained solution is given by ${\bf t}={\bf x}-{\rho}{\bm \varphi }$. By element-wise projection onto the constraint (\ref{tcons}), the optimal solution to problem (\ref{updt}) is,
\begin{equation}
\begin{aligned}
   & \Re\{{\bf t}_m\} =\min(\max(\Re\{{\bf x}-{\rho}{\bm \varphi }\},-\sqrt{\tfrac{1}{2M}}),\sqrt{\tfrac{1}{2M}}),\\
  & \Im\{{\bf t}_m\} =\min(\max(\Im\{{\bf x}-{\rho}{\bm \varphi }\},-\sqrt{\tfrac{1}{2M}}),\sqrt{\tfrac{1}{2M}}). \label{resultt}
\end{aligned}    
\end{equation}

\subsubsection{Update $\bm \phi$}
The subproblem related to update $\bm \phi$ is given by,
\begin{subequations}
    \begin{align}
       \min_{\bm \phi} \quad & \frac{1}{2\rho}\|{\bm \phi}-{\bm \theta}+{\rho}{\bm \psi }\|_2^2,\label{upobjphi}\\
       \text { s.t. }\quad &|{\phi}_{n}| = 1, \forall n. \label{upphic}
    \end{align}
    \label{pphi}%
\end{subequations}
Due to the constant modulus constraints in (\ref{upphic}), this problem is non-convex. However, since ${\bm \phi}$ has unit modulus entries, the quadratic term with respect to ${\bm \phi}$ in (\ref{upobjphi}) becomes constant. Thus, the problem (\ref{pphi}) reduces to,
\begin{equation}
       \max_{|{ \phi}_n|=1,\forall n} \quad  \Re\{({\bm \theta}-\rho{\bm \psi})^H{ \bm \phi}\}. \label{reforphi}
\end{equation}
The maximum in (\ref{reforphi}) is achieved when each element of ${\bm \phi}$ aligns with the corresponding element of the linear coefficient ${\bm \theta}-\rho{\bm \psi}$. Hence, the optimal solution is explicitly given by, 
\begin{equation}
    {\bm \phi} = \text{exp}(j\cdot \arg({\bm \theta}-\rho{\bm \psi})).\label{nowphi}
\end{equation}
Although problem (\ref{pphi}) is non-convex, we present the Lemma \ref{lemma34} regarding its optimality.

\begin{lemma} \label{lemma34}  \itshape
The solution given by (\ref{nowphi}) is the unique global optimum of problem (\ref{pphi}).
\end{lemma}

{\bf\textit{Proof}}: See Appendix \ref{Pro34}. $\hfill\blacksquare$

\underline{\textit{Outer Layer Procedure}}\\
\indent When the inner-loop BCD-based algorithm converges, we update either the dual variables $\{{\bm \varphi }, {\bm \psi}\}$ or the penalty parameter $\rho$ in the outer layer. Specifically, if the equalities ${\bf t} = {\bf x}$ and ${\bm \phi} = {\bm \theta}$ are approximately satisfied, i.e., both $\|{\bf t}-{\bf x}\|_{\infty}$ and $\|{\bm \phi}-{\bm \theta}\|_{\infty}$ are simultaneously smaller than a predefined diminishing threshold $\eta$ \cite{shi2020penalty}, we update the dual variables $\{{\bm \varphi }, {\bm \psi}\}$ to enhance the robustness of outer-loop convergence,
\begin{equation}
        {\bm \varphi } = {\bm \varphi } + \rho^{-1}({\bf t}-{\bf x}),\quad{\bm \psi} = {\bm \psi} + \rho^{-1}({\bm \phi}-{\bm \theta}).
\end{equation}
Otherwise, if the equality constraints ${\bf t} = {\bf x}$ and ${\bm \phi} = {\bm \theta}$ are not sufficiently satisfied, to enforce convergence in subsequent iterations, i.e., $\rho = c\cdot \rho$, where $c\in(0,1)$ is a predetermined positive constant.

\subsection{Overview of the WMMSE-PDD Algorithm}
Based on the above discussion, the complete WMMSE-PDD algorithm for solving problem (\ref{overallproblem2}) is summarized in Algorithm \ref{alg:1}. This algorithm is simple and efficient, as each step has an analyzed solution. Furthermore, as demonstrated subsequently, it converges to a KKT point.

\begin{algorithm}
	\floatname{algorithm}{Algorithm}
	\renewcommand{\algorithmicrequire}{\textbf{Input:}}
	\renewcommand{\algorithmicensure}{\textbf{Output:}}
	\caption{: WMMSE-PDD algorithm to solve (\ref{overallproblem2}).}
	\label{alg:1}
	\begin{algorithmic}[1]
		\REQUIRE 
                $w_b,w_e,{\bf v},{\bf x},{\bm \theta},{\bf t},{\bm \phi},  \rho,c,\eta, {\eta_{\min}},\varepsilon$, and ${\cal L}_{\rho}^{\text{new}}(\cdot)$.\\
            \STATE {\textbf {Repeat}} \textcolor{magenta}{\texttt{\%~Outer PDD loop}} 
            \STATE\quad {\textbf {Repeat}} \textcolor{magenta}{\texttt{\%~Inner BCD loop -- solve (\ref{Lagrangian})}} 
            \STATE\quad\quad ${\cal L}_{\rho}^{\text{old}}(\cdot) = {\cal L}_{\rho}^{\text{new}}(\cdot)$;
            \STATE\quad\quad  Update $w_b$ and $w_e$ by (\ref{upwewb});
            \STATE\quad\quad  Update $\bf v$ by (\ref{upv});
            \STATE\quad\quad  Update $\bf x$ by (\ref{upx});
            \STATE\quad\quad  Update $\bm \theta$ by (\ref{uptheta});
            \STATE\quad\quad  Update $\bf t$ by (\ref{updt});
             \STATE\quad\quad Update $\bm \phi$ by (\ref{nowphi});
             \STATE\quad\quad Update ${\cal L}_{\rho}^{\text{new}} (w_b,w_e,{\bf v},{\bf x},{\bm \theta},{\bf t},{\bm \phi};{\bm \varphi }, {\bm \psi} )$ by (\ref{Lagrangian}); 
            \STATE\quad {\textbf {Until}} $\frac{|{\cal L}_{\rho}^{\text{old}}(\cdot) - {\cal L}_{\rho}^{\text{new}}(\cdot)|}{|{\cal L}_{\rho}^{\text{old}}(\cdot)|}\le \varepsilon $
            \STATE\quad $\text{error} =\max\{\|{\bf t}-{\bf x}\|_{\infty},\|{\bm \phi}-{\bm \theta}\|_{\infty}\}$ 
            \STATE \quad \textbf{If} $\text{error} <\eta $ \textcolor{magenta}{\texttt{\%~Update dual variables}} 
            \STATE \quad \quad  ${\bm \varphi } = {\bm \varphi } + \rho^{-1}({\bf t}-{\bf x})$;
            \STATE \quad  $\quad{\bm \psi} = {\bm \psi} + \rho^{-1}({\bm \phi}-{\bm \theta})$;
            \STATE \quad \textbf{Else} \textcolor{magenta}{\texttt{\%~Update penalty parameter}} 
            \STATE \quad\quad $\rho = c\cdot \rho$;
            \STATE \quad \textbf{End if}
            \STATE\quad  ${\eta = 0.2 \times \text{error} }$, and ${\varepsilon = 0.1 \times \varepsilon }$
		\STATE {\textbf {Until } $\text{error}\le\eta_{\min}$.}\\
	\end{algorithmic}%
\end{algorithm}

\subsubsection{Analysis of computation complexity} To analyze the computational complexity, we focus on the primary contributors $M$ and $N_i$, since $N_e, N_b \ll M, N_i$. Specifically, updating $w_b$ and $w_e$ each has a complexity of approximately ${\cal O}(N_i^2 M)$, updating $\mathbf{v}$ has a complexity of approximately ${\cal O}(M)$, updating $\mathbf{x}$ has a complexity of approximately ${\cal O}(M^3)$, and updating $\bm{\theta}$ has a complexity of approximately ${\cal O}(N_i^3)$, while updating $\mathbf{t}$ and $\bm{\phi}$ has linear complexity, which can be neglected. Assuming the total number of inner-loop iterations is $T_{P_{\text{in}}}$ and the number of outer-loop iterations is $T_{P_{\text{out}}}$, the overall computational complexity of the WMMSE-PDD algorithm is approximately ${\cal O}(T_{P_{\text{out}}} T_{P_{\text{in}}} (N_i^3 + M^3))$.

\subsubsection{Analysis of Convergence} To establish convergence, we first demonstrate that the inner-loop BCD-based algorithm converges to a KKT point of problem (\ref{Lagrangian}), as stated in the following theorem.

\begin{theorem} \label{theorem31} \itshape
  Let ${\bf \mathcal Z}^k = \{w_b^k, w_e^k, {\bf{v}}^k, {\bf{x}}^k, {\bm{\theta}}^k, {\bf{t}}^k, {\bm{\phi}}^k\}$ denote the sequence generated by the inner-loop BCD-based algorithm. Then, every limit point of this sequence is a KKT point of problem (\ref{Lagrangian}).
\end{theorem}

\textit{Proof}: See Appendix \ref{appendixD}. $\hfill\blacksquare$

With the convergence of the inner-loop BCD-based algorithm to a KKT point established and the ease of satisfying Robinson’s condition \cite{kuhn2013nonlinear} for problem (\ref{overallproblem2}) shown, we further demonstrate that the proposed WMMSE-PDD algorithm also converges to a KKT point of problem (\ref{overallproblem2}), based on the following result from Theorem 3.1 in \cite{shi2020penalty}.

\begin{theorem} \label{theorem31} \itshape
Let ${\mathcal{Q}}^q = \{w_b^q, w_e^q, \mathbf{v}^q, \mathbf{x}^q, \bm{\theta}^q, \mathbf{t}^q, \bm{\phi}^q\}$ denote the sequence generated by the WMMSE-PDD algorithm in Algorithm \ref{alg:1}. Suppose that the inner-loop BCD-based algorithm converges to a KKT point. Let ${\mathcal{Q}}^*$ be a limit point of the sequence ${\mathcal{Q}}^q$, and assume that Robinson’s condition \cite{kuhn2013nonlinear} holds at ${\mathcal{Q}}^*$ for problem (\ref{overallproblem2}). Then, ${\mathcal{Q}}^*$ is a KKT point of problem (\ref{overallproblem2}).
\end{theorem}

\textit{Proof}: See Theorem 3.1 in \cite{shi2020penalty}. $\hfill\blacksquare$

\section{Proposed EPPRGD Algorithm}
Although the algorithm proposed in the previous section is highly efficient, its computational complexity increases due to the introduction of several auxiliary variables, which may slow convergence, especially when the number of antennas becomes extremely large. To enhance convergence speed, we further propose a lightweight EPPRGD algorithm based on advanced Riemannian manifold optimization. By avoiding the introduction of additional auxiliary variables, the EPPRGD algorithm significantly reduces complexity while maintaining satisfactory secrecy performance. The details are presented below.

\subsection{Problem Transformation}
To begin with, based on the one-bit transformation in (\ref{refor1}), and by converting the maximization to a minimization and swapping the numerator and denominator in problem (\ref{overallproblem2}), the problem (\ref{overallproblem2}) can be equivalently reformulated as,
\begin{subequations}
\begin{align}
 \min _{{\bf x}, {\bm \theta}} \quad& \log\frac{{1}+||({\mathbf H}_{ie} {\bf{\Theta }} {\mathbf H}_{ai}+{\mathbf H}_{ae} ) {\bf x}||_2^2}{{1}+||({\mathbf H}_{ib} {\bf{\Theta }} {\mathbf H}_{ai}+{\mathbf H}_{ab} ) {\bf x}||_2^2}, \label{objRe1}\\
  \text { s.t. }\quad &-\sqrt{\tfrac{1}{2M}}\le\Re\{x_m\},\Im\{x_m\}\le\sqrt{\tfrac{1}{2M}},\forall m,\label{Remacon1}\\
  &{\mathbf x}^H{\mathbf x} = 1,\label{Remacon2}\\
 & |{\theta}_{n}| = 1, \forall n. \label{Remacon3}
\end{align}
\label{Remaproblem1}%
\end{subequations}
Generally, constraints (\ref{Remacon2}) and (\ref{Remacon3}) are non-convex, and in the earlier stage, auxiliary variables were introduced to mitigate their difficulty. However, we also observe that these constraints can be interpreted as smooth manifolds over a Riemannian space, which motivates us to solve problem (\ref{Remaproblem1}) by designing the algorithm over these smooth Riemannian manifolds. Nevertheless, since the inequality constraint (\ref{Remacon1}) does not lie on any Riemannian manifold, this approach cannot be directly applied. To address this, we can incorporate it into the objective function as a penalty term to discourage violations. Based on these insights, in the following, an efficient EPPRGD algorithm is developed to solve the problem over the Riemannian space.

\subsection{Realization of the EPPRGD Algorithm}
The proposed EPPRGD algorithm adopts a double-loop structure, where the inner loop employs a PRGD algorithm to solve an unconstrained optimization problem over the Riemannian space. In the outer loop, the EP \cite{di1994exact} method is used to update the penalty parameters and smoothing variables, thereby enforcing satisfaction of the inequality constraint. The details are presented as follows.

Since the inequality constraint (\ref{Remacon1}) does not lie on any Riemannian manifold, manifold-based optimization methods cannot be directly applied. To address this issue, we adopt the EP \cite{di1994exact} method to incorporate constraint (\ref{Remacon1}) into the objective function. Specifically, the exact penalty method replaces the inequality constraint with a weighted penalty term, thereby augmenting the objective to penalize constraint violations. After applying this method, the problem is reformulated as follows,
\begin{equation}
\begin{aligned}
 \min _{{\bf x}, {\bm \theta}} \quad& f_r({\bf x}, {\bm \theta})+{\rho_r}  \textstyle\sum\limits_{m=1}^{M} \left(\begin{array}{l}  
 \max\{0,g_1(\Re\{x_m\}) \}\\+\max\{0,g_1(\Im\{x_m\}) \}\\+\max\{0,g_2(\Re\{x_m\}) \}\\+\max\{0,g_2(\Im\{x_m\}) \}\end{array}\right), \\
 \text {s.t.} \quad&(\ref{Remacon2})\text{ and } (\ref{Remacon3}) \text{ are satisfied}, \label{Reforcon2}
\end{aligned}
\end{equation}
where $f_r({\bf x}, {\bm \theta})$ is the objective function (\ref{objRe1}), $\rho_r>0$ is the penalty parameter, and,
\begin{subequations}
\begin{align}
& g_1(\Re\{x_m\})=\Re\{x_m\}-\sqrt{\tfrac{1}{2M}},\\
&g_1(\Im\{x_m\})=\Im\{x_m\}-\sqrt{\tfrac{1}{2M}},\\
& g_2(\Re\{x_m\})=-\Re\{x_m\}-\sqrt{\tfrac{1}{2M}},\\
& g_2(\Im\{x_m\})=-\Im\{x_m\}-\sqrt{\tfrac{1}{2M}}.
\end{align}
\label{totalett}
\end{subequations}
Notice that the resulting penalized objective function in (\ref{Reforcon2}) remains non-smooth, which may pose challenges for optimization. Fortunately, the cost function in (\ref{Reforcon2}) consists of a sum of maximum functions, a structure that can be exploited algorithmically. Specifically, we adopt the log-sum-exp function \cite{boyd2004convex} to obtain a smooth approximation, as described in the Lemma \ref{lemma44}.

\begin{lemma} \label{lemma44} \itshape
\textit{(Section 3.1.5 in \cite{boyd2004convex})} Given $c_1,c_2,..., c_K \in \mathbb{R}$, it holds for any smoothing parameter $u > 0$ that,
\begin{equation}
    \max_{k=1,...,K} c_k \le u \log \sum_{k=1}^K \exp\left({\frac{c_k}{u}}\right) \le \max_{k=1,...,K} c_k + u \log K.
\end{equation}
\end{lemma}
\textit{Moreover, the inequalities become tight as $u \to 0$}.

Based on Lemma \ref{lemma44}, we can approximate the non-smooth penalty terms. Specifically, by setting $K=2$, $c_1=0$, and $c_2=x$ in the lemma, we derive the smooth approximation for the hinge loss function $\max\{0, x\}$ as follows,
\begin{equation}
\begin{aligned}
   {{ \max\{0, x\} }} {{\le u \log ( \exp(\frac{0}{u}) + \exp(\frac{x}{u}) ) }}
     {{= u \log ( 1 + \exp(\frac{x}{u}) ).}}
\end{aligned}
\end{equation}
Applying this approximation to each penalty term in problem (\ref{Reforcon2}), the optimization problem is smoothly reformulated as,
\begin{equation}
\begin{aligned}
  {\min _{{\bf x}, {\bm \theta}}}\quad&  {{f_r({\bf x}, {\bm \theta})+{\rho_r} u \textstyle\sum\limits_{m=1}^{M} \left(\begin{array}{l}  
 \log(1+e^{\frac{g_1(\Re\{x_m\})}{u}})  \\+\log(1+e^{\frac{g_1(\Im\{x_m\})}{u}})\\+\log(1+e^{\frac{g_2(\Re\{x_m\})}{u}})\\+\log(1+e^{\frac{g_2(\Im\{x_m\})}{u}})\end{array}\right), }}\\
 {{ \text {s.t.} }}\quad& {{(\ref{Remacon2})\text{ and } (\ref{Remacon3}) \text{ are satisfied}.}} \label{Reforcon3}
\end{aligned}
\end{equation}
After reformulation, problem (\ref{Reforcon3}) involves only the manifold-structured constraints (\ref{Remacon2}) and (\ref{Remacon3}). Based on the definition of Riemannian manifolds \cite{lee2018introduction}, constraints (\ref{Remacon2}) and (\ref{Remacon3}) correspond to the complex sphere manifold $\mathcal{M}_{\bf x}$ and the complex circle manifold $\mathcal{M}_{\bm \theta}$, respectively, as given by,
\begin{equation}
    \begin{aligned}
       & \mathcal{M}_{\bf x} = \left\{{\bf x} \in \mathbb{C}^{N_t } \mid \|{\bf x}\|_2 = 1 \right\},\\
       & \mathcal{M}_{{\bm \theta}} = \left\{ {\bm \theta} \in \mathbb{C}^{{N_i}} \mid |\theta_n|=1,\forall n \in N_i \right\}.\label{man1}
    \end{aligned}
\end{equation}
The Cartesian product of multiple individual Riemannian manifolds can form a product Riemannian manifold \cite{boumal2023introduction}. Based on this definition, the manifolds \( \mathcal{M}_{\bf x} \) and $\mathcal{M}_{{\bm \theta}}$ can be combined to form a product Riemannian manifold \( \mathcal{M} \), given by,
\begin{equation}
    \mathcal{M} = \mathcal{M}_{\bf x} \times \mathcal{M}_{{\bm \theta}}. \label{man2}
\end{equation}
To enable optimization over the Riemannian space, the tangent space is required as it provides a local linear approximation of the manifold. For the product Riemannian manifold $\mathcal{M}$, the corresponding tangent space is defined as \cite{boumal2023introduction},
\begin{equation}
{\mathcal T}{\mathcal M} = {\mathcal T}_{\bf x}\mathcal{M}_{\bf x} \oplus {\mathcal T}_{{\bm \theta}}\mathcal{M}_{{\bm \theta}}, \label{tanspcae}
\end{equation}
where \( {\mathcal T}_{\bf x}\mathcal{M}_{\bf x} \) and \( {\mathcal T}_{{\bm \theta}}\mathcal{M}_{{\bm \theta}} \) are the tangent spaces of the individual manifolds \( \mathcal{M}_{\bf x} \) and \( \mathcal{M}_{{\bm \theta}} \), respectively, defined as \cite{boumal2023introduction},
\begin{equation}
    \begin{aligned}
       & {\mathcal T}_{\bf x}\mathcal{M}_{\bf x} = \{ {\bf \Xi}\in\mathbb{C}^{M } \mid \Re\{{\bm \zeta }^H{\bf x}\}= 0 \},\\
       & {\mathcal T}_{{\bm \theta}}\mathcal{M}_{{\bm \theta}} = \{ {\bm \theta}\in\mathbb{C}^{N_{i}} \mid \Re\{{\bm \omega }^*  \odot {\bm \theta}\}= 0 \},\label{tangentres}
    \end{aligned}
\end{equation}
where $\bm{\zeta}$ and $\bm{\omega}$ are tangent vectors within their respective tangent spaces. Based on (\ref{man2})–(\ref{tangentres}), and letting $\bm{\Upsilon} = \mathbf{w} \oplus \bm{\theta}$ denote the collection of the variable set $\{\mathbf{w}, \bm{\theta}\}$, problem (\ref{Reforcon3}) can be reformulated as an unconstrained optimization problem over the product Riemannian manifold $\mathcal{M}$, given as,
\begin{equation}
    \min_{\bm{\Upsilon}\in\mathcal{M}} g_r(\bm{\Upsilon}),\label{manpro3}
\end{equation}
where $g_r(\bm{\Upsilon})$ represents the objective function in the problem (\ref{Reforcon3}). By fixing the penalty parameter $\rho_r$ and the smoothing parameter $u$, the unconstrained optimization problem (\ref{manpro3}) over the Riemannian space can be solved using the widely adopted Riemannian gradient descent (RGD) method \cite{boumal2023introduction}. In the following, we propose the PRGD algorithm based on the RGD framework to solve the problem (\ref{manpro3}), which forms the inner-loop algorithm.

\underline{\textit{Inner Layer Procedure}}\\
\indent In general, the PRGD algorithm extends classical gradient descent to Riemannian manifolds by computing the Riemannian gradient within the tangent space and then mapping the updated point back onto the manifold through a retraction operation. This approach preserves the geometric structure of the manifold while ensuring iterative progress toward convergence. Specifically, the PRGD algorithm cyclically performs the following three steps: (1) calculation of the Riemannian gradients, (2) adaptive adjustment of the step sizes, and (3) update of the feasible solutions, for \(k = 0, 1, \ldots\), until a stopping criterion is satisfied.

\subsubsection{Calculation of Riemannian gradient}
The Riemannian gradient determines the descent direction in the PRGD algorithm. Specifically, the Riemannian gradient $\text{grad}_{\mathcal{M}}\, g_r(\bm{\Upsilon})$ with respect to the variable $\bm{\Upsilon}$ is obtained by projecting its Euclidean gradient onto the tangent space ${\mathcal T}{\mathcal M}$, given by,
\begin{equation}
\begin{aligned}
      & \text{grad}_{\mathcal{M}} g_r({\bf \Upsilon})= \left\{\begin{array}{l}  \text{grad}_{\mathcal{M}_{\bf x}}g_r({\bf \Upsilon})\\\oplus\text{grad}_{\mathcal{M}_{{\bm \theta}}}g_r({\bf \Upsilon}) \end{array} \right\}  \\ &=\left\{\begin{array}{l}  (\nabla_{{\bf x}} g_r({\bf \Upsilon})-\Re\{\nabla_{{\bf x}}^H g_r({\bf \Upsilon}){\bf x}\}{\bf x})\\ 
      \oplus(\nabla_{{\bm \theta}} g_r({\bf \Upsilon})-\Re\{\nabla^*_{{\bm \theta}}g_r({\bf \Upsilon})\odot {\bm \theta} \}\odot{\bm \theta}) \end{array} \right\},\label{calculReg}
\end{aligned}       
\end{equation}
where \( \text{grad}_{\mathcal{M}_{\bf x}}g_r({\bf \Upsilon}) \) and \( \text{grad}_{\mathcal{M}_{{\bm \theta}}}g_r({\bf \Upsilon})  \) are the Riemannian gradients associated with the individual manifolds \( \mathcal{M}_{\bf x} \) and \( \mathcal{M}_{{\bm \theta}} \), respectively. \( \nabla_{{\bf x}} g_r({\bf \Upsilon}) \) and \( \nabla_{{\bm \theta}} g_r({\bf \Upsilon}) \) denote the corresponding Euclidean gradients, given by (\ref{eua1}) and (\ref{eua2}), respectively,
\begin{subequations}
\begin{align}
    &\nabla_{{\bf x}} f({\bf \Upsilon}) =  2({\bf H}_{e}^H({\bm \theta}){\bf Q}_{e}-{\bf H}_{b}^H({\bm \theta}){\bf Q}_{b}){\bf x}\notag\\&+\rho\left\{\begin{array}{l}[\kappa(\tfrac{g_1(\Re\{x_m\})}{u})-\kappa(\tfrac{g_2(\Re\{x_m\})}{u})]\\+j[\kappa(\tfrac{g_1(\Re\{x_m\})}{u})-\kappa(\tfrac{g_2(\Re\{x_m\})}{u})] \end{array}\right\},\label{eua1}\\
    &\nabla_{{\bm \phi}_1}f({\bf \Upsilon})= 2\text{diag}(({\bf H}_{ie}^H{\bf Q}_{e}-{\bf H}_{ib}^H{\bf Q}_{b}){\bf x}{\bf x}^H{\bf H}_{a{i}}^H),\label{eua2}
\end{align}    
\end{subequations}    
where $\kappa(\cdot) \triangleq \tfrac{e^{({\cdot})}}{1+e^{({\cdot})}}$, ${\bf Q}_{e} = \tfrac{{\bf H}_{e}(\bm \theta)}{1 +  \|{\bf H}_{e}({\bm \theta}){\bf x}\|_2^2}$ and ${\bf Q}_{b} = \tfrac{{\bf H}_{b}(\bm \theta)}{1 +  \|{\bf H}_{b}({\bm \theta}){\bf x}\|_2^2}$.

\subsubsection{Determination of the Step Size}
Utilizing the Armijo linear search strategy \cite{boyd2004convex}, we dynamically adjust the step size during updates. This adaptive approach aligns the search step size with variations, preserving the property of non-increasing objective function values and enhancing algorithmic convergence speed. In particular, the chosen step size \( \alpha^k \) has to satisfy,
\begin{equation}
\mathcal{L}({\bf \Upsilon}^{k+1}) -\mathcal{L}({\bf \Upsilon}^{k})\le -\frac{1}{2}\alpha^k \|\text{grad}_{\mathcal{M}} \mathcal{L}({\bf \Upsilon}^k)\|_F^2. \label{stepsize}
\end{equation}

\subsubsection{Update and Retraction}
With the Riemannian gradient obtained from (\ref{calculReg}) and the step size determined by (\ref{stepsize}), the update in \( \mathcal{T} \mathcal{M} \) is given by,
\begin{equation}
    {  \bf \hat \Upsilon} = {\bf \Upsilon}^{k} - \alpha^k \text{grad}_{\mathcal{M}} f({\bf \Upsilon}^k).\label{Updateupsi}
\end{equation}
Since product Riemannian manifold \( \mathcal{M} \) is nonlinear, updates may result in \( {\bf \hat \Upsilon} \) falling outside of \( \mathcal{M} \). To address this, a retraction operation $\text{Ret}(\cdot)$ is applied to map the updated point back onto \( \mathcal{M} \), as given by,
\begin{equation}
   {\bf \Upsilon}^{k+1}= \text{Ret}( {  \bf \hat \Upsilon})= \text{Ret}( {  \bf \hat w}\oplus{ \bm {\hat \theta}})=({\bf \hat w}\oslash{\|{\bf \hat w}\|_2^2}) \oplus ( {\bm {\hat \theta}} \oslash |{\bm {\hat \theta}}|). \label{useforretrac}
\end{equation}

\underline{\textit{Outer Layer Procedure}}\\
\indent When the inner-loop PRGD algorithm converges, we update both the penalty parameter $\rho_r$ and the smoothing parameter $u$ in the outer loop. Specifically, if the maximum violation $\max\{g_1(\Re\{x_m\}), g_1(\Im\{x_m\}), g_2(\Re\{x_m\}), g_2(\Im\{x_m\}), \forall m\}$ is approximately satisfied, i.e., all terms are simultaneously smaller than a predefined diminishing threshold $\tau$, it indicates that the current penalty parameter $\rho_r$ is sufficient to eliminate the violation, and thus $\rho_r$ is kept unchanged. Otherwise, we increase $\rho_r$ to further enforce the satisfaction of the constraint, according to $\rho_r = c_r \cdot \rho_r$, where $c_r > 1$ is a predetermined positive constant. Concurrently, the smoothing parameter is reduced following $u = 0.2 \times u$ to progressively narrow the gap between the reformulated problem (\ref{Reforcon3}) and the original problem (\ref{Reforcon2}). This gradual decay strategy circumvents the numerical intractability of directly initializing with a small $u$, which would induce excessive landscape stiffness as the Lipschitz constant of the gradient scales inversely with the parameter \cite{nesterov2005smooth}. Instead, the proposed path-following approach enables the solver to warm-start effectively from the previous iteration, thereby ensuring stability \cite{grossmann2016smoothing}.

\subsection{Overview of the EPPRGD Algorithm}
Based on the above discussion, the complete EPPRGD algorithm for solving problem (\ref{overallproblem2}) is summarized in Algorithm \ref{alg:2}. Due to the benefit of not introducing auxiliary variables, the EPPRGD algorithm in Algorithm \ref{alg:2} achieves faster convergence than Algorithm \ref{alg:1}, as will be shown later in the numerical results section. Moreover, we also prove that it achieves KKT convergence.

\begin{algorithm}
	\floatname{algorithm}{Algorithm}
	\renewcommand{\algorithmicrequire}{\textbf{Input:}}
	\renewcommand{\algorithmicensure}{\textbf{Output:}}
	\caption{: EPPRGD algorithm to solve (\ref{overallproblem2}).}
	\label{alg:2}
	\begin{algorithmic}[1]
		\REQUIRE 
                ${\bf x},{\bm \theta},  \rho_r,c_r,\tau, u, \textcolor{blue}{u_{{\min}}},\varepsilon_r$, and $g_r^{\text{new}}(\cdot)$.\\
            \STATE {\textbf {Repeat}} \textcolor{magenta}{\texttt{\%~Outer EPPRGD loop}} 
            \STATE\quad Combine ${\bm \Upsilon}={\bf x} \oplus  {\bm \theta}$
            \STATE\quad {\textbf {Repeat}} \textcolor{magenta}{\texttt{\%~Inner PRGD loop -- solve (\ref{manpro3})}} 
            \STATE\quad\quad $g_r^{\text{old}}(\cdot) = g_r^{\text{new}}(\cdot)$;
            \STATE\quad\quad  Update Riemannian gradient $\text{grad}_{\mathcal{M}} g_r({\bf \Upsilon})$ by (\ref{calculReg});
            \STATE\quad\quad  Update step size $\alpha$ by (\ref{stepsize});
            \STATE\quad\quad  Update next feasible ${\bm \Upsilon}$ by (\ref{useforretrac});
             \STATE\quad\quad Update $g_r^{\text{new}}(\cdot)$ by (\ref{manpro3}); 
            \STATE\quad {\textbf {Until}} $\frac{|g_r^{\text{old}}(\cdot)-g_r^{\text{new}}(\cdot)|}{|g_r^{\text{old}}(\cdot)|}\le \varepsilon_r $
            \STATE\quad Decompose ${\bf \Upsilon}$ to obtain ${\bf x}$ and ${\bm \theta}$;
            \STATE\quad $\text{error}_r =\max\left\{\begin{array}{l} g_1(\Re\{x_m\}), g_1(\Im\{x_m\}),\\g_2(\Re\{x_m\}), g_2(\Im\{x_m\})\end{array}\right\},\forall m$
            \STATE \quad \textbf{If} $\text{error}_r \le\tau $ \textcolor{magenta}{\texttt{\%~Penalty is enough}} 
            \STATE \quad \quad  $\rho_r=\rho_r$;
            \STATE \quad \textbf{Else} \textcolor{magenta}{\texttt{\%~Penalty is not enough}} 
            \STATE \quad\quad $\rho_r=c_r\cdot\rho_r$;
            \STATE \quad \textbf{End if}
            \STATE\quad  ${u = 0.2 \times u }$, and ${\varepsilon_r = 0.1 \times \varepsilon_r }$
		\STATE {\textbf {Until } $\text{error}_r\le \tau$ and $u\le u_{\min}$.}\\
	\end{algorithmic}%
\end{algorithm}

\subsubsection{Analysis of computation complexity} 
Similar to the previous analysis, we focus on the main parameters $M$ and $N_i$. In general, the complexity of the EPPRGD algorithm is primarily dominated by the calculation of the Riemannian gradients $\text{grad}_{\mathcal{M}_{\mathbf{x}}}g_r(\bm{\Upsilon})$ and $\text{grad}_{\mathcal{M}_{\bm{\theta}}}g_r(\bm{\Upsilon})$. Specifically, the computation of $\text{grad}_{\mathcal{M}_{\mathbf{x}}}g_r(\bm{\Upsilon})$ has a complexity of approximately ${\cal O}(M^2N_i)$, while the computation of $\text{grad}_{\mathcal{M}_{\bm{\theta}}}g_r(\bm{\Upsilon})$ has a complexity of approximately ${\cal O}(N_i^2M)$. Assuming the total number of inner-loop iterations is $T_{R_{\text{in}}}$ and the number of outer-loop iterations is $T_{R_{\text{out}}}$, the overall computational complexity of the EPPRGD algorithm is approximately ${\cal O}(T_{R_{\text{out}}} T_{R_{\text{in}}} N_iM(N_i + M))$.

\subsubsection{Joint complexity analysis and hardware efficiency} Having presented the two proposed optimization algorithms, namely WMMSE-PDD in Section III and EPPRGD in this section, we now provide a unified analysis of their computational complexity and hardware efficiency. This comparison highlights the practical advantages of our proposed designs over the conventional SDR-based benchmark.

We utilize computational complexity as a theoretical proxy for FPGA resource consumption (e.g., digital signal processing (DSP) slices and logic units) \cite{meyer2007digital}. As summarized in Table \ref{tab:complexity}, the conventional SDR-based method incurs a prohibitive complexity order of $\mathcal{O}((M+N_i)^{6})$, typically requiring iterative interior-point methods that are computationally expensive for hardware implementation. In contrast, both of our proposed algorithms significantly reduce the computational burden compared to the SDR benchmark. The WMMSE-PDD algorithm involves matrix inversions with a complexity of $\mathcal{O}((M + N_i)^3)$. Similarly, the EPPRGD algorithm relies primarily on gradient calculations, with a complexity of $\mathcal{O}(M^2 N_i + MN_i^2)$. Both approaches achieve a substantial reduction in computational cost compared to the $\mathcal{O}((M+N_i)^{6})$ complexity of traditional SDR methods. This low-complexity characteristic directly translates to fewer required logic gates and lower power consumption on FPGA platforms, validating the scalability of both proposed schemes for practical deployment.

Furthermore, it is worth noting that although the computational complexity of each inner iteration in the proposed EPPRGD algorithm is similar to that of the WMMSE-PDD algorithm, we observed that it requires significantly fewer inner iterations, greatly enhancing the convergence speed, as demonstrated in Fig. \ref{Con_1} and Fig. \ref{Con_2}. Additionally, the slight performance degradation of the EPPRGD algorithm is attributed to the larger number of penalty terms (approximately $4M$ terms) compared to only $M + N_i$ terms in the WMMSE-PDD algorithm. This increase in penalty terms imposes a greater challenge in balancing convergence speed with the achievable secrecy rate.

\begin{table}[!t]
\centering
\caption{Theoretical Complexity Analysis (FPGA Proxy)}
\label{tab:complexity}
\renewcommand{\arraystretch}{1.2}
\begin{tabular}{|l|l|l|}
\hline
\textbf{Algorithm} & \textbf{Complexity Order} & \textbf{Dominant Operation} \\ \hline
SDR-IRS (Benchmark) & $\mathcal{O}((M+N_i)^{6})$ & Interior Point Method \\ \hline
\textbf{WMMSE-PDD} & $\mathcal{O}((M + N_i)^3)$ & Matrix Inversion \\ \hline
\textbf{EPPRGD} & $\mathcal{O}(M^2 N_i + N_i^2 M)$ & Gradient Calculation \\ \hline
\end{tabular}
\end{table}

\subsubsection{Analysis of Convergence}
To establish convergence, we first show that the inner-loop PRGD algorithm converges to a KKT point, as stated in the following theorem.

\begin{theorem} \label{theorem3}  \itshape
Let $\{\bm{\Upsilon}^k\}=\{{\bf w}^k,{\bm \theta}^k\}$ be a sequence generated by the PRGD algorithm. Then, every limit point of this sequence is a KKT point of problem (\ref{manpro3}).
\end{theorem}

\textit{Proof}: See Appendix \ref{appendixC}. $\hfill\blacksquare$

Building on the convergence of the inner-loop PRGD algorithm and the establishment of the linear independence constraint qualification (LICQ) condition for problem (\ref{overallproblem2}), we further demonstrate that the proposed EPPRGD algorithm in Algorithm \ref{alg:2} converges to a KKT point of problem (\ref{overallproblem2}), as established by the following theorem based on Proposition 4.2 in \cite{liu2020simple}.

\begin{theorem} \label{theorem2} \itshape
Let $\{\bm{\Upsilon}^k\}$ be the sequence generated by Algorithm~\ref{alg:2}. Assume that, at each outer iteration, the inner-loop PRGD yields a KKT point of the subproblem (as established in Theorem~\ref{theorem3}), and that the initial penalty parameter $\rho_r^0$ is sufficiently large. Furthermore, consider the limiting case where $u_{\min} \to 0$, ensuring that the smoothing sequence $\{u^k\}$ vanishes. If $\bm{\Upsilon}^*$ is a feasible limit point of $\{\bm{\Upsilon}^k\}$ at which the LICQ holds \cite{kuhn2013nonlinear}, then $\bm{\Upsilon}^*$ is a KKT point of problem~(\ref{overallproblem2}).
\end{theorem}

\textit{Proof}: See Appendix \ref{appendixF}. $\hfill\blacksquare$

\section{Numerical Results}
In this section, we provide simulation results to evaluate the performance of the SRM design for the proposed IRS-assisted mMIMO system. For comparison, the proposed method is compared to the following benchmark schemes:
\begin{itemize}
\item IRS ($\infty$ bits): The IRS-assisted secure precoding design in \cite{asaad2022secure}, which assumes an infinite-resolution precoder. Since the one-bit constraint is relaxed, this method serves as the upper bound.
\item SDR-IRS (one bit): This approach alternately optimizes (AO) the one-bit precoding vector and the IRS phase shifts using the semi-definite relaxation (SDR) methods in \cite{wei2020transmit} and \cite{cui2019secure}, respectively. While conceptually straightforward, this scheme has not been previously reported.
\item DP-IRS (one bit): This method first applies the design in 'IRS ($\infty$ bits)', followed by naive one-bit quantization of the precoder via direct projection.
\item WOIRS ($\infty$ bits): This baseline follows the secure precoding design in \cite{8714023}, where the precoding vector is optimized with infinite resolution without IRS assistance.
\item WOIRS (one bit): Based on the one-bit secure precoding method in \cite{loukil2023physical}, this scheme does not incorporate IRS assistance.
\end{itemize}

Following \cite{asaad2022secure,kong2024distributed,xu2021artificial}, we model the channels as comprising both large-scale and small-scale fading components. The large-scale path loss is given by $L(d) = C_0 (d/D_0)^{-\nu }$, where \( d = \sqrt{(x_r - x_t)^2+(y_r - y_t)^2 } \) denotes the distance between the transmitter and receiver, with \( (x_t, y_t) \) and \( (x_r, y_r) \) representing their respective locations. Here, \( \nu \) is the path loss exponent, and \( C_0 = -30 \) dB represents the path loss at the reference distance \( D_0 = 1 \) m.

For the direct channels ${\bf \hat H}_{ae}$ and ${\bf \hat H}_{ab}$, the small-scale fading is modeled as a complex Gaussian distribution due to extensive scattering. However, for the IRS-related channels \( \mathbf{\hat H}_{ai} \), \( \mathbf{\hat H}_{ib} \), and \( \mathbf{\hat H}_{ie} \), small-scale fading follows a Rician distribution, given by,    
\begin{equation}
    \mathbf{\tilde{H}} = \sqrt{\tfrac{\vartheta }{1+\vartheta }} \mathbf{\tilde{H}}^{\text{LOS}} + \sqrt{\tfrac{1}{1+\vartheta }} \mathbf{\tilde{H}}^{\text{NLOS}},
\end{equation}  
where \( \mathbf{\tilde{H}}^{\text{LOS}} = \mathbf{a}_r(\iota_r) \mathbf{a}_t^H(\iota_t) \) represents the line-of-sight (LoS) component, where \( \mathbf{a}_r(\iota_r) \) and \( \mathbf{a}_t(\iota_t) \) are the steering vectors of the receive and transmit arrays, respectively, expressed as \( \mathbf{a}_j = [1, e^{j\pi\sin(\iota_j)}, \dots, e^{j\pi(M_j-1)\sin(\iota_j)}], \forall j \in \{r,t\} \), with the angles of departure (AoD) and arrival (AoA) given by \( \iota_t = \tan^{-1}(\frac{y_r - y_t}{x_r - x_t}) \) and \( \iota_r = \pi - \iota_t \), respectively. \(\mathbf{\tilde{H}}^{\text{NLOS}} \) is the non-line-of-sight (NLOS) component follows a complex Gaussian distribution. 

Unless otherwise specified, the parameters are set as follows: Alice, Bob, and Eve are equipped with \( M = 128 \), \( N_b = 16 \), and \( N_e = 16 \) antennas, respectively. The IRS consists of \( N_i = 256 \) reflecting coefficients. The total transmit power is \( P = 30 \) dBm, while the noise power at Bob and Eve is \( \sigma_b^2 = \sigma_e^2 = -50 \) dBm. 

The path loss exponents for the channels \( \mathbf{\hat H}_{ai} \), \( \mathbf{\hat H}_{ib} \), \( \mathbf{\hat H}_{ie} \), \( \mathbf{\hat H}_{ab} \), and \( \mathbf{\hat H}_{ae} \) are set to \( \nu_{ai} = 2.2 \), \( \nu_{ib} = 2.5 \), \( \nu_{ie} = 2.5 \), \( \nu_{ab} = 3.5 \), and \( \nu_{ae} = 3.5 \), respectively. The Rician factor is \( \vartheta  = 5 \). The coordinates of Alice, IRS, Eve, and Bob are given as (0,0) m, (50,0) m, (45,2) m, and (55,2) m, respectively. All results are averaged over 5,00 independent realizations and obtained on a PC with a 3.8 GHz Intel Core Ultra 7 155H processor.

Figs. \ref{Con_1} and \ref{Con_2} illustrate the convergence behavior of the WMMSE-PDD and EPPRGD algorithms presented in Algorithms \ref{alg:1} and \ref{alg:2}, respectively. Firstly, Figs. \ref{Con_1}(a) and \ref{Con_2}(a) show the outer-loop convergence characteristics of these algorithms. Both algorithms exhibit a monotonic increase in secrecy rate and a monotonic decrease in the maximum violation as the iteration number grows. Moreover, convergence for both methods is attained within six iterations, with the maximum violation error and $\text{error}_r$ falling below $10^{-5}$. Notably, the WMMSE-PDD algorithm achieves a higher secrecy rate of 7.947 bps/Hz compared to approximately 7.384 bps/Hz attained by the EPPRGD algorithm, underscoring the performance advantage of the WMMSE-PDD algorithm. Furthermore, the inner-loop behaviors of the WMMSE-PDD and EPPRGD algorithms are presented in Figs. \ref{Con_1}(b) and \ref{Con_2}(b), respectively. Both algorithms demonstrate monotonic decreases within each inner loop. It can also be observed that the objective function value at the end of each inner loop is smaller than that at the start of the subsequent inner loop. The observed increase in the objective function value at the beginning of each subsequent inner loop arises due to increments in the penalty term. Additionally, the EPPRGD algorithm requires significantly fewer inner-loop iterations (50 iterations) to achieve overall convergence compared with the WMMSE-PDD algorithm (200 iterations), highlighting the computational efficiency advantage of the EPPRGD algorithm.

\begin{figure}[t]
  \begin{center}
  \includegraphics[width=3in]{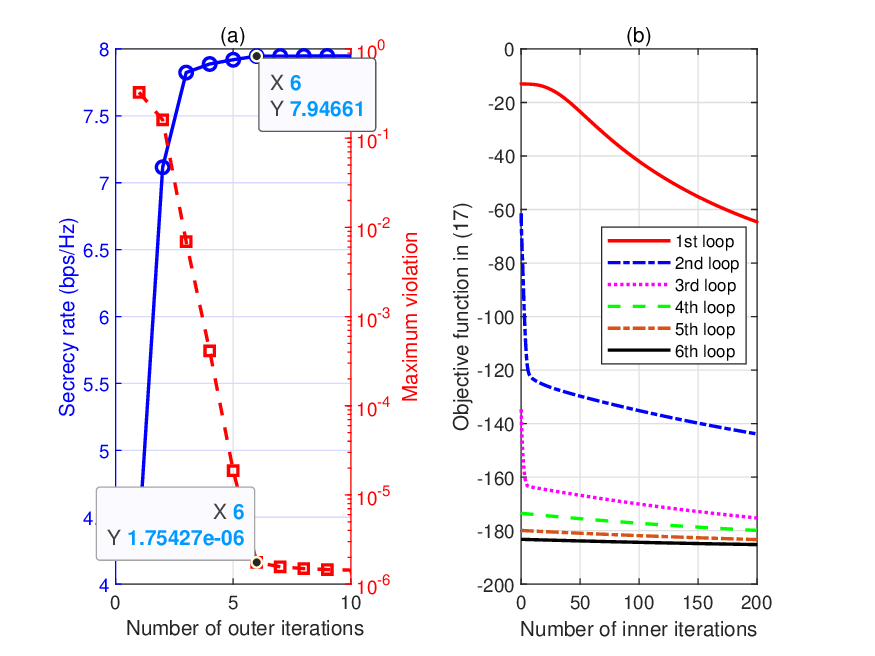}\\
  \caption{Convergence performance of the WMMSE-PDD algorithm.}\label{Con_1}
  \end{center}
\end{figure}

\begin{figure}[t]
  \begin{center}
  \includegraphics[width=3in]{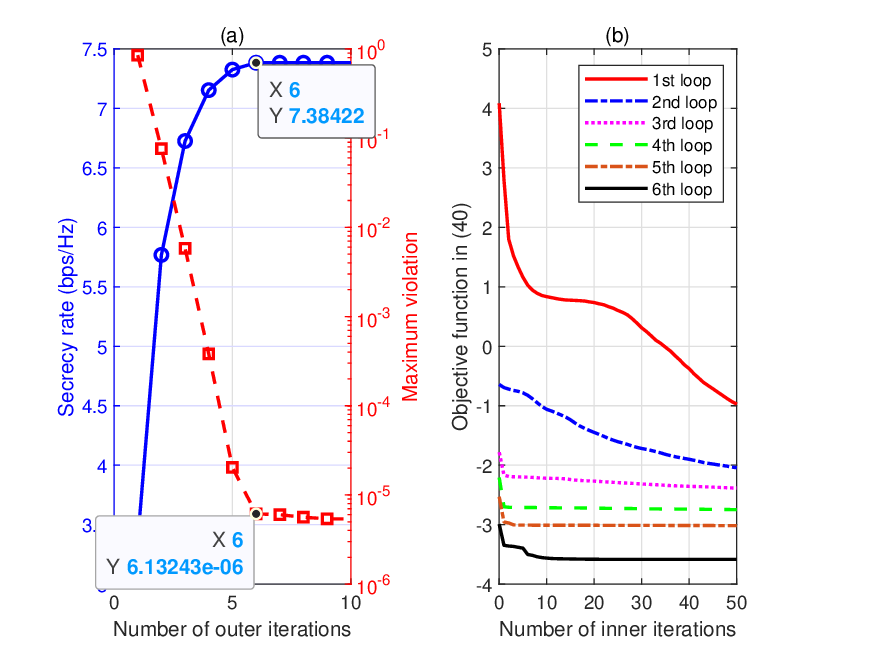}\\
  \caption{Convergence performance of the EPPRGD algorithm.}\label{Con_2}
  \end{center}
\end{figure}

\begin{figure}[t]
  \begin{center}
  \includegraphics[width=2.5in]{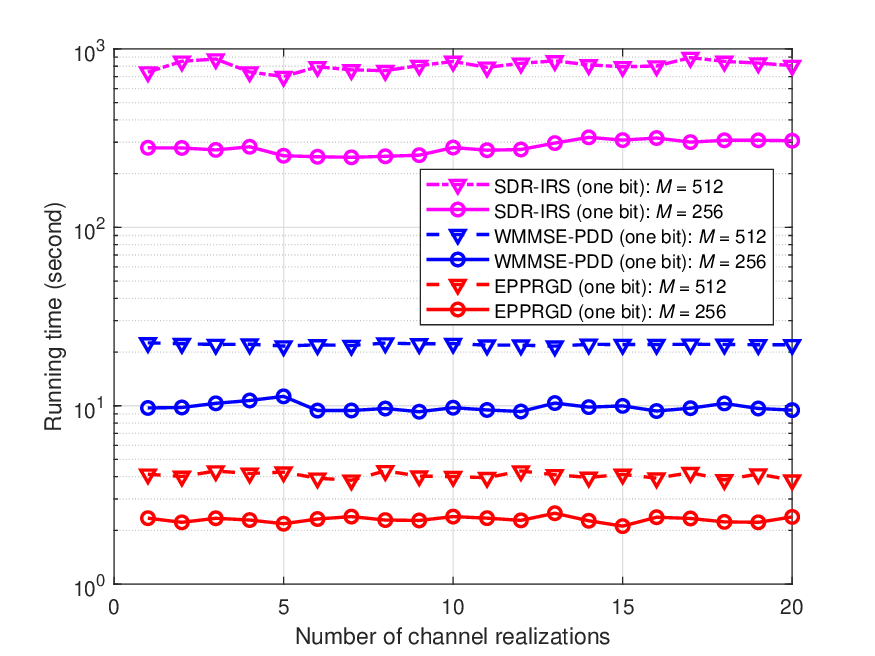}\\
  \caption{Comparison of empirical runtime among architectures on the first twenty randomly generated channels.}\label{Dif_time}
  \end{center}
\end{figure}

Fig. \ref{Dif_time} compares the total empirical runtime required to reach convergence over first twenty independent experimental trials, where each trial corresponds to a complete simulation run with a randomly generated channel realization. The proposed EPPRGD algorithm consistently achieves the fastest convergence, with an average total runtime of 2.30 s per run for $M=256$, which is approximately 2.7× faster than the proposed WMMSE-PDD algorithm (6.24 s). This advantage further increases to nearly 3.9× for $M=512$ (4.07 s versus 15.84 s). In contrast, the SDR-IRS (one-bit) benchmark incurs prohibitively high computational cost, with average runtimes of 282.59 s ($M=256$) and 808.25 s ($M=512$). To decouple convergence speed from algorithmic complexity, we also examine the per-iteration computation time, where EPPRGD requires only 0.0077 s ($M=256$) and 0.014 s ($M=512$) per iteration, slightly outperforming WMMSE-PDD (0.0082 s and 0.018 s) and vastly surpassing SDR-IRS (14.13 s and 40.41 s). This substantial reduction in both total runtime and per-iteration cost highlights the practical efficiency of the proposed algorithms, enabled by closed-form updates and their inherently distributed implementation.

\begin{figure}[t]
  \begin{center}
  \includegraphics[width=2.5in]{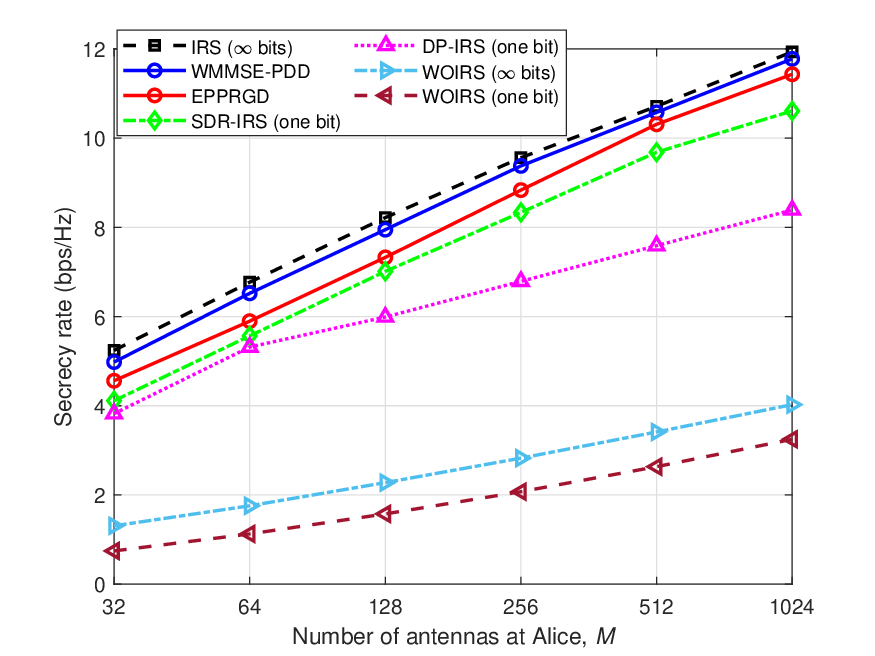}\\
  \caption{Comparison of secrecy rates among architectures as the
number of antennas at Alice ($M$) increases from 32 to 1024.}\label{Dif_M}
  \end{center}
\end{figure}

Fig. \ref{Dif_M} illustrates the secrecy rate performance of various architectures as the number of antennas at Alice, denoted by $M$, varies. As depicted, the secrecy rate consistently increases with a larger $M$ for all considered architectures. This performance gain arises from the additional spatial degrees of freedom introduced by increasing the antenna number, highlighting the advantage of employing mMIMO systems for secure communications. Moreover, both the proposed WMMSE-PDD and EPPRGD algorithms achieve higher secrecy rates compared to the SDR-IRS (one bit), DP-IRS (one bit), WOIRS ($\infty$ bits), and WOIRS (one bit) algorithms. These results indicate that the proposed methods are particularly effective for one-bit precoding design in mMIMO systems with large antenna arrays. Specifically, at $M = 128$, the secrecy rates attained by the WMMSE-PDD and EPPRGD algorithms are approximately 7.95 bps/Hz and 7.33 bps/Hz, respectively, outperforming the SDR-IRS (one bit) algorithm (7.01 bps/Hz), DP-IRS (one bit) algorithm (5.99 bps/Hz), WOIRS ($\infty$ bits) algorithm (2.28 bps/Hz), and WOIRS (one bit) algorithm (1.57 bps/Hz). Furthermore, the proposed algorithms approach the secrecy performance of the upper-bound IRS ($\infty$ bits) algorithm, demonstrating their effectiveness and efficiency.

\begin{figure}[t]
  \begin{center}
  \includegraphics[width=2.5in]{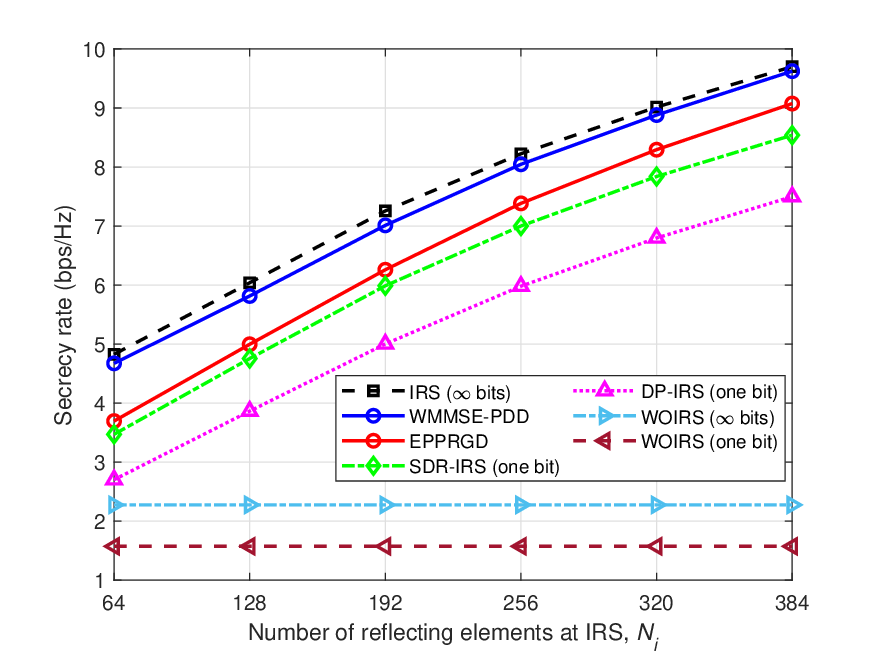}\\
  \caption{Comparison of secrecy rates between architectures as the
number of reflecting elements at the IRS ($N_i$) increases from 64 to 384.}\label{Dif_Irs}
  \end{center}
\end{figure}

Fig. \ref{Dif_Irs} compares the secrecy rate performance of various architectures with different numbers of IRS reflecting elements $N_i$. As observed, the secrecy rate improves with increasing $N_i$ for all architectures employing an IRS. This enhancement demonstrates that the IRS provides additional spatial degrees of freedom, substantially improving overall system performance. Specifically, at $N_i = 256$, the proposed WMMSE-PDD and EPPRGD algorithms achieve secrecy rates of approximately 8.05 bps/Hz and 7.38 bps/Hz, respectively, surpassing those of the SDR-IRS (one bit), DP-IRS (one bit), WOIRS ($\infty$ bits), and WOIRS (one bit) algorithms, which achieve secrecy rates of approximately 7.00 bps/Hz, 5.98 bps/Hz, 2.28 bps/Hz, and 1.57 bps/Hz, respectively. These results clearly demonstrate the effectiveness and hardware efficiency of the proposed algorithms in the design of one-bit DAC-based precoding for IRS-assisted mMIMO systems.

\begin{figure}[t]
  \begin{center}
  \includegraphics[width=2.5in]{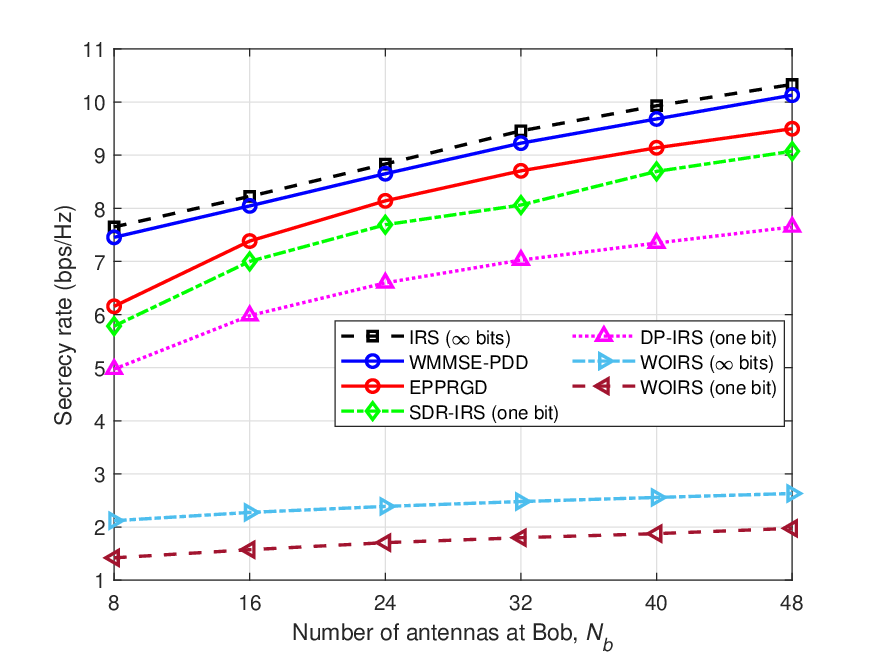}\\
  \caption{Comparison of secrecy rates among architectures as the
number of antennas at Bob ($N_b$) increases from 8 to 48.}\label{Dif_Nb}
  \end{center}
\end{figure}

\begin{figure}[t]
  \begin{center}
  \includegraphics[width=2.5in]{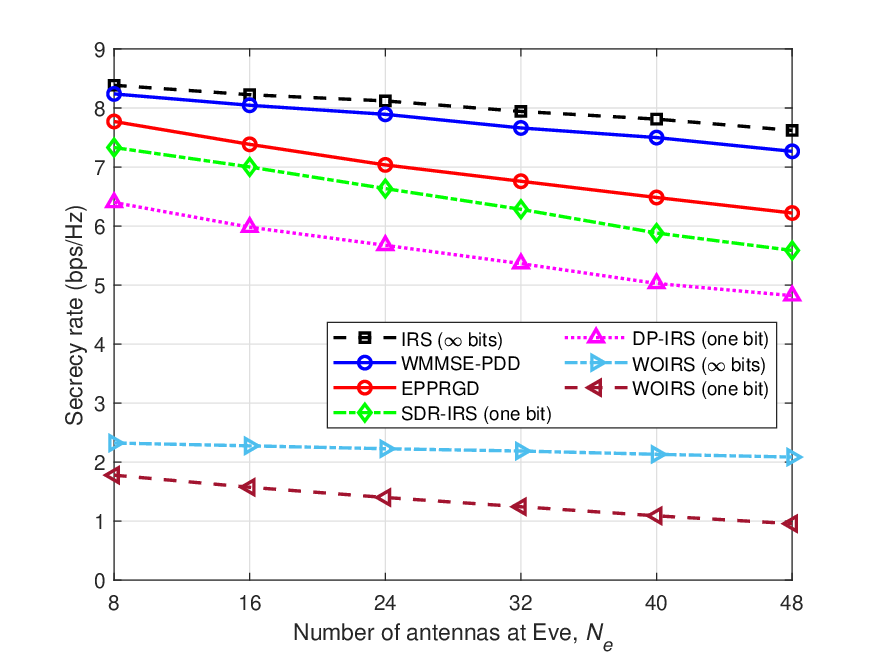}\\
  \caption{Comparison of secrecy rates among architectures as the
number of antennas at Eve ($N_e$) increases from 8 to 48.}\label{Dif_Ne}
  \end{center}
\end{figure}

To evaluate robustness against CSI imperfections, we adopt the channel estimation error model \cite{11314715}. For a given realization of Eve's true channel $\mathbf{H} \in \{\mathbf{H}_{ae}, \mathbf{H}_{ie}\}$, the estimated channel is modeled as $\hat{\mathbf{H}} = \mathbf{H} + \mathbf{E}$, where the entries of $\mathbf{E}$ are i.i.d. variables distributed as $\mathcal{CN}(0, \sigma_E^2)$. We quantify the error severity using the normalized mean square error (NMSE), denoted by $\delta_e$. The variance is determined by the instantaneous channel power as $\sigma_E^2 = \frac{\delta_e \|\mathbf{H}\|_F^2}{K_e}$, where $K_e$ is the total number of elements in $\mathbf{H}$. During simulations, optimization is performed using $\hat{\mathbf{H}}$, while actual secrecy rates are evaluated using the true $\mathbf{H}$. Fig. \ref{Dif_CUN} illustrates the secrecy rate performance against the channel uncertainty level of Eve, $\delta_e$. As observed, the secrecy rate declines with increasing $\delta_e$ for all architectures, attributable to the CSI mismatch that weakens the signal suppression at Eve. Despite this degradation, the proposed algorithms exhibit superior robustness. Specifically, at $\delta_e = 0.5$, the proposed WMMSE-PDD and EP-PRGD algorithms maintain secrecy rates of approximately 6.30 bps/Hz and 5.52 bps/Hz, respectively, significantly outperforming the SDR-IRS (one bit), DP-IRS (one bit), and WOIRS (one bit) benchmarks, which achieve 4.95 bps/Hz, 3.90 bps/Hz, and 0.28 bps/Hz, respectively. Notably, the WMMSE-PDD algorithm approaches the performance of the ideal IRS ($\infty$ bits) upper bound (6.75 bps/Hz), demonstrating the effectiveness of the proposed framework in mitigating CSI imperfections.

\begin{figure}[t]
  \begin{center}
  \includegraphics[width=2.5in]{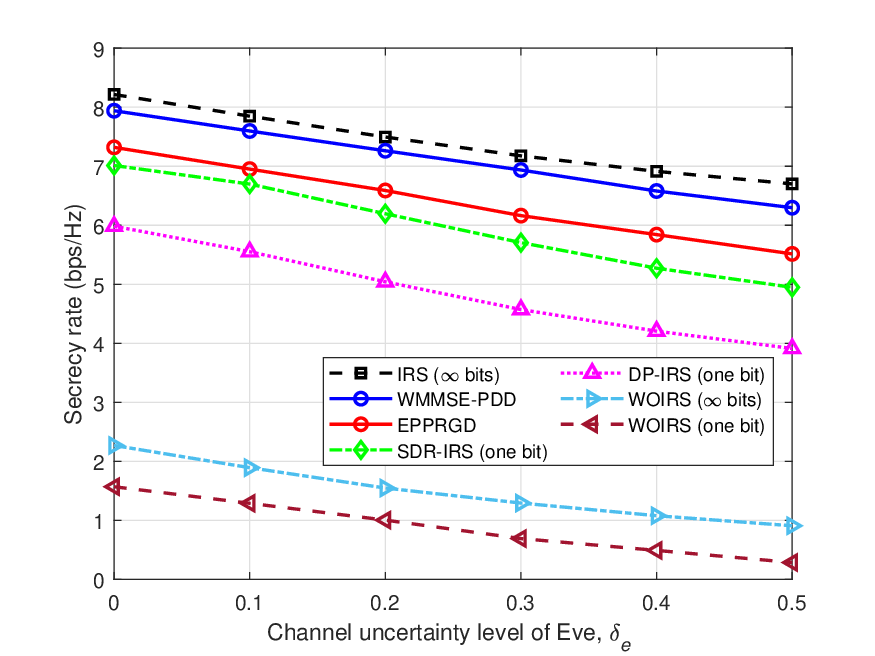}\\
  \caption{Comparison of secrecy rates among architectures as the channel uncertainty level of Eve ($\delta_e$) increases from 0 to 0.5.}\label{Dif_CUN}
  \end{center}
\end{figure}

To demonstrate the extensibility of our design to frontier multi-IRS scenarios, we consider the distributed multi-IRS system in \cite{rafieifar2023secure}, to which our proposed frameworks can be directly applied. Specifically, we consider a deployment with two distributed IRSs located at $(10, 0)$ m and $(50, 0)$ m, respectively. To ensure a fair comparison with the centralized single-IRS benchmark ($N_i=256$), the total number of reflecting elements is kept constant by setting $N_{i_1}=N_{i_2}=128$ for the distributed IRSs, while maintaining all other parameters unchanged. Fig. \ref{dis_mul} depicts the secrecy rate versus the total transmit power $P$. It is observed that the multi-IRS architectures significantly outperform their single-IRS counterparts across the entire power regime. For instance, at $P=45$ dBm, the proposed WMMSE-PDD algorithm in the multi-IRS setup achieves a secrecy rate of approximately 19.5 bps/Hz, providing a substantial gain of about 6 bps/Hz over the single-IRS case (13.5 bps/Hz). This gain stems from the distributed spatial architecture, which increases the probability of favorable line-of-sight links and provides superior degrees of freedom to effectively suppress the signal at Eve.

\begin{figure}[t]
  \begin{center}
  \includegraphics[width=2.5in]{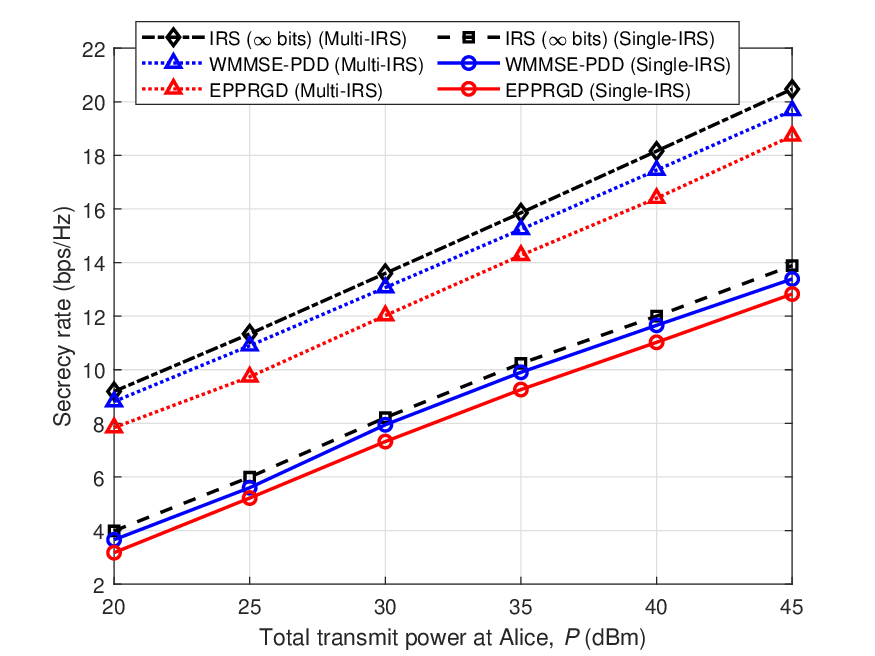}\\
  \caption{Comparison of secrecy rates between single-IRS and distributed multi-IRS architectures as the total transmit power ($P$) increases from 20 to 45 dBm.}\label{dis_mul}
  \end{center}
\end{figure}

\section{Conclusion}
This paper investigated a cost-effective PLS design for IRS-assisted mMIMO systems with one-bit DACs. A joint optimization framework was developed to maximize the secrecy rate by designing one-bit precoding and IRS phase shifts. Two algorithms were proposed: WMMSE-PDD, achieving superior secrecy performance, and EPPRGD, enabling faster convergence by eliminating auxiliary variables. Both algorithms provided analytical updates and were proven to converge to KKT points. Simulations validated the effectiveness of the proposed methods. Future work will extend the framework to scenarios with imperfect CSI and dynamic IRS control.

\begin{appendices} 

\section{Proof of Lemma~\ref{lemma31}}
\label{MPDD1}

Let $a \triangleq \sqrt{1/(2M)}$ and define the discrete one-bit alphabet as $\mathcal{X} = \{ \pm a \pm ja \}$. We establish the equivalence between the discrete set $\mathcal{X}^M$ and the continuous constraints in (\ref{refor1}) through the following analysis of equality conditions.

\subsubsection{Proof of Necessity}  
Assume $x_m \in \mathcal{X}$ for all $m \in \{1, \dots, M\}$. By the definition of $\mathcal{X}$, it holds that $\mathfrak{R}\{x_m\} \in \{ -a, a \}$ and $\mathfrak{I}\{x_m\} \in \{ -a, a \}$, which directly implies the box constraints,
\begin{equation}
-a\le\mathfrak{R}\{x_m\} \le a, \quad -a\le\mathfrak{I}\{x_m\} \le a, \quad \forall m.
\end{equation}
Moreover, the squared magnitude of each element is given by,
\begin{equation}
|x_m|^2 = \mathfrak{R}\{x_m\}^2 + \mathfrak{I}\{x_m\}^2 = a^2 + a^2 = \frac{1}{M}.
\end{equation}
Summing the squared magnitudes over all $M$ antennas, we obtain,
\begin{equation}
\mathbf{x}^H \mathbf{x} = \sum_{m=1}^M |x_m|^2 = \sum_{m=1}^M \frac{1}{M} = 1.
\end{equation}

\subsubsection{Proof of Sufficiency} 
Conversely, assume that the iterate $\mathbf{x}$ satisfies $\mathbf{x}^H \mathbf{x} = 1$ and the element-wise constraints $|\mathfrak{R}\{x_m\}| \le a, |\mathfrak{I}\{x_m\}| \le a$ for all $m$. Under these box constraints, the squared magnitude of any entry $x_m$ is bounded by,
\begin{equation}
|x_m|^2 = \mathfrak{R}\{x_m\}^2 + \mathfrak{I}\{x_m\}^2 \le a^2 + a^2 = \frac{1}{M}.
\end{equation}
Summing these individual bounds across the entire vector yields the following inequality for the total power,
\begin{equation}
\mathbf{x}^H \mathbf{x} = \sum_{m=1}^M |x_m|^2 \le \sum_{m=1}^M \frac{1}{M} = 1.
\end{equation}
According to the premise, the equality $\mathbf{x}^H \mathbf{x} = 1$ must hold. For the sum of $M$ non-negative terms to strictly reach its maximum upper bound, each individual term must attain its respective maximum value,
\begin{equation}
|x_m|^2 = \frac{1}{M}, \quad \forall m \in {1, \dots, M}.
\end{equation}
Given that $|x_m|^2 = \mathfrak{R}\{x_m\}^2 + \mathfrak{I}\{x_m\}^2$ and considering the individual constraints $\mathfrak{R}\{x_m\}^2 \le a^2$ and $\mathfrak{I}\{x_m\}^2 \le a^2$, the identity $|x_m|^2 = 2a^2$ necessitates,
\begin{equation}
\mathfrak{R}\{x_m\}^2 = a^2 \quad \text{and} \quad \mathfrak{I}\{x_m\}^2 = a^2.
\end{equation}
This forces $\mathfrak{R}\{x_m\} \in \{ \pm a \}$ and $\mathfrak{I}\{x_m\} \in \{ \pm a \}$, which implies $x_m \in \mathcal{X}$ for all $m$. This completes the proof. \hfill $\blacksquare$

\section{Proof of Lemma \ref{lemma33}}
\label{Pro33}
Suppose, on the contrary, that another optimal solution $\tilde{\mathbf{x}}\neq\mathbf{x}$ exists. Optimality implies that there is a multiplier $\tilde{\lambda}$ such that $(\tilde{\mathbf{x}},\tilde{\lambda})$
satisfies KKT conditions,
\begin{subequations}\label{kkt}
\begin{align}
&(2\rho\mathbf{A}+(1+2\rho\tilde{\lambda})\mathbf{I}_M)\tilde{\mathbf{x}}
      =\mathbf{t}+\rho\bm{\varphi}+2\rho\mathbf{b}, \label{kkt_grad}\\
&\tilde{\mathbf{x}}^H\tilde{\mathbf{x}}=1,
\end{align}
\end{subequations}
whence $\tilde{\mathbf{x}}=\mathbf{x}(\tilde{\lambda})$ in (\ref{upx}) and
$\mathbf{x}^H(\tilde{\lambda})\mathbf{x}(\tilde{\lambda})=1$. Given $\lambda$ is unique, $\tilde{\lambda}$ must equal $\lambda$, forcing $\tilde{\mathbf{x}}=\mathbf{x}(\lambda)=\mathbf{x}$, contradicting $\tilde{\mathbf{x}}\neq\mathbf{x}$. Therefore, no other distinct minimizer exists; $\mathbf{x}$ is indeed unique. This completes the proof. $\hfill\blacksquare$

\section{Proof of Lemma \ref{lemma34}}
\label{Pro34}
We prove the claim by contradiction. Suppose there exists another feasible solution $\widetilde{\bm\phi}\neq\bm\phi$ achieving the same objective value.  
Since the terms $\Re\{(\theta_n-\rho\psi_n)^*\phi_n\}$ are additive and separable across $n$, and each $\phi_n$ is constrained to lie on the unit circle, equality of the overall objective requires that,
\begin{equation}
\Re\{(\theta_n-\rho\psi_n)^*\widetilde{\phi}_n\} = \Re\{(\theta_n-\rho\psi_n)^*\phi_n\}, \quad \forall n.
\end{equation}
Given that the maximum real part for each $n$ is achieved only when $\widetilde{\phi}_n = \exp(j\arg(\theta_n-\rho\psi_n))$, it follows that $\widetilde{\bm\phi}=\bm\phi$. This contradicts the assumption $\widetilde{\bm\phi}\neq\bm\phi$. Therefore, $\bm\phi$ is the unique global minimizer. This completes the proof. $\hfill\blacksquare$


\section{Convergence of the BCD Algorithm}
\label{appendixD}
To begin with, let ${\bf \mathcal Z}^{k} \triangleq \{w_b^{k}, w_e^{k}, {\bf{v}}^{k}, {\bf{x}}^{k}, {\bm{\theta}}^{k}, {\bf{t}}^{k}, {\bm{\phi}}^{k}\}$ denote the variables at iteration $k$. The convergence relies on three key properties of the BCD updates,

     \textit{1) Objective Function Descent and Lower Bound:} In each BCD step, ${\cal L}_{\rho}(\mathcal{Z})$ is minimized with respect to a single variable block while keeping the others fixed. Since each subproblem is solved to its unique global optimum (as detailed below), we have ${\cal L}_{\rho}(\mathcal{Z}^{k+1}) \leq {\cal L}_{\rho}(\mathcal{Z}^{k})$, indicating that the objective function is monotonically decreasing over iterations. Furthermore, we show that the objective ${\cal L}_{\rho}(\mathcal{Z})$ is bounded below. Specifically, $w_b$ and $w_e$ are strictly positive due to the $-\log(\cdot)$ terms and their update rules in (\ref{upwewb}); $\mathbf{x}$, $\mathbf{t}$, and $\bm{\phi}$ are constrained within compact sets as given in (\ref{pdd2})–(\ref{pdd4}); the penalty terms are non-negative; and $\mathbf{v}$ and $\bm{\theta}$ remain bounded due to the quadratic terms in their subproblems and the overall descent of the objective. Therefore, the sequence $\{{\cal L}_{\rho}(\mathcal{Z}^{k})\}$ is non-increasing and bounded below, and thus convergent.

    \textit{2) Unique Global Optimum for Each Subproblem:} In the proposed iterative framework, every update step guarantees a unique global optimum. Specifically, the auxiliary variables $w_b, w_e$, and $\mathbf{v}$ are obtained as the unique solutions to the convex subproblems in (\ref{upwewb}) and (\ref{upv}), respectively. Similarly, the phase shift $\bm{\theta}$ and auxiliary variable $\mathbf{t}$ correspond to the unique global optima of the convex problems (\ref{uptheta}) and (\ref{updt}), with closed-form solutions given by (\ref{resulttheta}) and (\ref{resultt}). Regarding the primal variables, Lemma~\ref{lemma33} establishes that the update for $\mathbf{x}$ in (\ref{upx}) constitutes the unique global optimum of (\ref{upxnow1}). Likewise, as proven in Lemma~\ref{lemma34}, the update for $\bm{\phi}$ in (\ref{nowphi}) represents the unique global optimum of (\ref{pphi}) (equivalent to (\ref{reforphi})). Crucially, this uniqueness property ensures that the algorithm follows a strict descent path; that is, if the current iterate $\mathcal{Z}^{k}$ is not a KKT point, the augmented Lagrangian strictly decreases, i.e., ${\cal L}_{\rho}(\mathcal{Z}^{k+1}) < {\cal L}_{\rho}(\mathcal{Z}^{k})$.

     \textit{3) Compactness of Iterates:} We show that the sequence $\mathcal{Z}^{k}$ is contained within a compact set. To begin with, the variables $\bf{x}, \bf{t}, \bm{\phi}$ are explicitly constrained to compact sets. Moreover, the boundedness of $w_b, w_e, \bf{v}, \bm{\theta}$ follows from the objective descent property and the coercive nature of their respective subproblem objectives.

Given that ${\cal L}_{\rho}(\mathcal{Z})$ is continuous, each subproblem yields a unique global minimum, and the iterates $\mathcal{Z}^{k}$ lie in a compact set, established convergence results for BCD methods (e.g., Theorem 4.1 in \cite{shi2020penalty}) ensure that any limit point of $\mathcal{Z}^{k}$ is a KKT point of (\ref{Lagrangian}). This completes the proof. $\hfill\blacksquare$

\section{Convergence of the PRGD Algorithm}\label{appendixC}
To establish the proof, we first show that the objective function in (\ref{manpro3}) achieves sufficient decrease at each iteration. This is ensured by carefully choosing the step size in (\ref{stepsize}), i.e.,
\begin{equation}
\begin{aligned}
g_r(\bm{\Upsilon}^k) - g_r(\bm{\Upsilon}^{k+1}) \ge c_{\text{dec}} \| \operatorname{grad}_{\mathcal{M}} g_r(\bm{\Upsilon}^k) \|_2^2, \label{therem4proof}
\end{aligned}
\end{equation}
where $c_{\text{dec}} =\tfrac{1}{2}\alpha^k >0$. We can now complete the proof. The proof is based on a standard telescoping sum argument. The desired inequality for all $k=0,1,\dots,K-1$ is obtained as follows,
\begin{subequations}
    \begin{align}
        g_r(\bm{\Upsilon}^0) - &g_r^{\text{low}} \ge g_r(\bm{\Upsilon}^0) - g_r(\bm{\Upsilon}^K) \\ & = \sum\limits_{k=0}^{K-1} g_r(\bm{\Upsilon}^k) - g_r(\bm{\Upsilon}^{k+1}) \\& \ge K c_{\text{dec}} \min _{k=0,1,\dots,K-1}\| \operatorname{grad}_{\mathcal{M}} g_r(\bm{\Upsilon}^k) \|_2^2
    \end{align}
\end{subequations}
where $g_r^{\text{low}}$ is the lower bound value for the objective function in (\ref{manpro3}). To get the limit statement, observe that $g_r(\bm{\Upsilon}^{k+1}) \le g_r(\bm{\Upsilon}^k)$ for all $k$ by (\ref{therem4proof}). Then, taking $K$ to infinity we see that,
\begin{equation}
    g_r(\bm{\Upsilon}^0) -  g_r^{\text{low}} \ge \sum\limits_{k=0}^{\infty} g_r(\bm{\Upsilon}^k) - g_r(\bm{\Upsilon}^{k+1}),
\end{equation}
where the right-hand side is a series of nonnegative numbers. The bound implies that the summands converge to zero, thus,
\begin{equation}
\begin{aligned}
    0 &= \lim_{k \to \infty} g_r(\bm{\Upsilon}^k) - g_r(\bm{\Upsilon}^{k+1}) \\ &\le c_{\text{dec}} \lim_{k \to \infty}\| \operatorname{grad}_{\mathcal{M}} g_r(\bm{\Upsilon}^k) \|_2^2,
    \end{aligned}
\end{equation}
which confirms that $\|\operatorname{grad}_{\mathcal{M}} g_r(\bm{\Upsilon}^k) \|_2\to0$. Now, let ${\bm{\Upsilon}}^*$ be a limit point of the sequence of iterates. By definition, there exists a subsequence of iterates
${\bm{\Upsilon}}^{(0)},{\bm{\Upsilon}}^{(1)},{\bm{\Upsilon}}^{(2)},\dots$ which converges to ${\bm{\Upsilon}}^*$. Then, since the norm of the gradient of $ g_r(\bm{\Upsilon})$ is a continuous function, it commutes with the limit and we find,
\begin{equation}
\begin{aligned}
    0 &= \lim_{k \to \infty}\| \operatorname{grad}_{\mathcal{M}} g_r(\bm{\Upsilon}^k) \|_2^2 = \lim_{k \to \infty}\| \operatorname{grad}_{\mathcal{M}} g_r(\bm{\Upsilon}^{(k)}) \|_2^2 \\ & = \|\operatorname{grad}_{\mathcal{M}} g_r(\lim_{k \to \infty}( \bm{\Upsilon}^{(k)})) \|_2^2 = \|\operatorname{grad}_{\mathcal{M}} g_r({\bm{\Upsilon}}^*) \|_2^2,
\end{aligned}
\end{equation}
showing that all limit points generated by the sequence are KKT points. This completes the proof. $\hfill\blacksquare$

\section{Convergence of the EPPRGD Algorithm}
\label{appendixF}

Let $\{{\bm \Upsilon}^k\}$ denote the sequence of iterates generated by Algorithm~\ref{alg:2}, where ${\bm \Upsilon} = {\bf x} \oplus {\bm\theta}$ represents the variables on $\mathcal{M} = \mathcal{M}_{\bf x} \times \mathcal{M}_{\bm \theta}$. At the $k$-th outer iteration, the objective $g_r(\cdot)$ in (\ref{manpro3}) is given by,
\begin{equation}
g_r({\bm \Upsilon}^k) = f_r({\bm \Upsilon}^k) + \rho_r^k u^k \sum_{i \in \mathcal{I}} \log( 1 + \exp(\frac{c_i({\bm \Upsilon}^k)}{u^k})),
\end{equation}
where $f_r(\cdot)$ is the original objective function in (\ref{objRe1}), and $\{c_i({\bm\Upsilon})\}_{i \in \mathcal{I}}$ denotes the set of inequality constraints in (\ref{totalett}).

We define the smoothing coefficient $\gamma_i^k$ associated with the $i$-th inequality constraint as,
\begin{equation}\label{eq:gamma_def_appE}
\gamma_i^k \triangleq \frac{\exp(c_i({\bm \Upsilon}^k)/u^k)}{1+\exp(c_i({\bm \Upsilon}^k)/u^k)} \in (0,1).
\end{equation}
Accordingly, the Riemannian gradient of $g_r(\cdot)$ is expressed as,
\begin{equation}\label{eq:eff_mult_structure_appE}
\operatorname{grad}_{\mathcal{M}} g_r({\bm \Upsilon}^k)=\operatorname{grad}_{\mathcal{M}} f_r({\bm \Upsilon}^k)+\sum_{i\in\mathcal{I}}\nu_i^k\operatorname{grad}_{\mathcal{M}} c_i({\bm \Upsilon}^k),\end{equation}
where $\nu_i^k \triangleq \rho_r^k \gamma_i^k$ serves as the effective multiplier.

Since $\mathcal{M}$ is compact, the sequence $\{{\bm \Upsilon}^k\}$ is bounded. Consequently, by the Bolzano--Weierstrass theorem \cite{oman2017short}, there exists a subsequence $\{{\bm \Upsilon}^{k_\ell}\}$ that converges to a limit point ${\bm \Upsilon}^*$. We assume that ${\bm \Upsilon}^*$ is feasible for problem~(\ref{overallproblem2}) and that the linear independence constraint qualification (LICQ) holds at ${\bm \Upsilon}^*$. Furthermore, Step~16 of Algorithm~\ref{alg:2} updates the smoothing parameter as $u^{k+1}=0.2 u^k$. Given the termination condition $u^k \le u_{\min}$ and the limit $u_{\min}\to 0$, the sequence $\{u^k\}$ is monotonically decreasing and approaches zero, which is essential for establishing complementary slackness. We now begin our proof.

{\it{1) Convergence of Multipliers:}} From \eqref{eq:eff_mult_structure_appE}, the effective multiplier to $c_i({\bm \Upsilon})$ is,
\begin{equation}\label{eq:multiplier_def}
\nu_i^k \triangleq \rho_r^k \gamma_i^k,
\end{equation}
where $\gamma_i^k \in (0,1)$ implies $0\le \nu_i^k \le \rho_r^k$.

According to Algorithm~\ref{alg:2}, the penalty parameter $\rho_r$ is increased only when the constraint violation $\text{error}_r$ exceeds a prescribed threshold. Since $\text{error}_r$ is a continuous function of the constraint residuals $\{c_i(\cdot)\}$ and ${\bm \Upsilon}^{k_\ell}\to{\bm \Upsilon}^*$ with $c_i({\bm \Upsilon}^*)\le 0$, we have $\text{error}_r({\bm \Upsilon}^{k_\ell})\to 0$. Hence, the penalty update can be triggered only finitely many times along $\{k_\ell\}$, which implies that $\{\rho_r^{k_\ell}\}$ is bounded. Therefore, $\{\rho_r^{k_\ell}\}$ admits a convergent subsequence; for notational simplicity, we relabel it still as $\{\rho_r^{k_\ell}\}$ and denote its limit by $\rho_r^*<\infty$.

Consequently, $\{\nu_i^{k_\ell}\}$ is bounded and thus admits a convergent subsequence. Without loss of generality, we assume,
\begin{equation}
\nu_i^{k_\ell} \to \nu_i^*, \quad \text{as } \ell \to \infty,
\end{equation}
where $\nu_i^* \ge 0$ holds by construction.

\textit{2) Complementary Slackness:} We verify $\nu_i^*\, c_i({\bm \Upsilon}^*) = 0$. Partition $\mathcal{I}$ into the active set $\mathcal{A}_{\text{active}}=\{i\mid c_i({\bm \Upsilon}^*)=0\}$ and the inactive set $\mathcal{I}_{\text{inactive}}=\{i\mid c_i({\bm \Upsilon}^*)<0\}$.

For $i \in \mathcal{A}_{\text{active}}$, $\nu_i^*\, c_i({\bm \Upsilon}^*)=0$ holds trivially.

For $i \in \mathcal{I}_{\text{inactive}}$, strict feasibility holds: $c_i({\bm \Upsilon}^*)<0$. Since each $c_i(\cdot)$ is continuous and ${\bm \Upsilon}^{k_\ell} \to {\bm \Upsilon}^*$, there exist $\delta>0$ and $\ell_0$ such that for all $\ell>\ell_0$,
\begin{equation}
c_i({\bm \Upsilon}^{k_\ell}) \le -\delta < 0.
\end{equation}
As $u^{k_\ell} \to 0$, the argument of the exponential term diverges to negative infinity,
\begin{equation}
\frac{c_i({\bm \Upsilon}^{k_\ell})}{u^{k_\ell}}
\le \frac{-\delta}{u^{k_\ell}}
\xrightarrow{\ell \to \infty} -\infty.
\end{equation}
Thus, $\gamma_i^{k_\ell}\to 0$ according to (\ref{eq:gamma_def_appE}). Since $\{\rho_r^{k_\ell}\}$ is bounded, we obtain $\nu_i^{k_\ell}=\rho_r^{k_\ell}\gamma_i^{k_\ell}\to 0$, which implies $\nu_i^*=0$. Therefore, $\nu_i^*\, c_i({\bm \Upsilon}^*)=0$ holds for all $i\in\mathcal{I}_{\text{inactive}}$, proving complementary slackness.

\textit{3) Stationarity:} Finally, we establish stationarity. Define the KKT residual vector ${\bm v} \in T_{{\bm \Upsilon}^*}\mathcal{M}$ at the limit point as,
\begin{equation}
    {\bm v} \triangleq \operatorname{grad}_{\mathcal{M}} f_r({\bm \Upsilon}^*) 
    + \sum_{i \in \mathcal{I}} \nu_i^* \operatorname{grad}_{\mathcal{M}} c_i({\bm \Upsilon}^*).
\end{equation}
To prove $\|{\bm v}\|_F = 0$, we introduce parallel transport to map tangent vectors from $T_{{\bm \Upsilon}^{k_\ell}}\mathcal{M}$ to $T_{{\bm \Upsilon}^*}\mathcal{M}$.
Since ${\bm \Upsilon}^{k_\ell} \to {\bm \Upsilon}^*$, for sufficiently large $\ell$, the iterate falls within the injectivity radius of the limit point \cite{absil2008optimization},
\begin{equation}
    \operatorname{dist}({\bm \Upsilon}^{k_\ell}, {\bm \Upsilon}^*) < \operatorname{inj}({\bm \Upsilon}^*).
\end{equation}
This condition guarantees the existence of a unique minimizing geodesic connecting ${\bm \Upsilon}^{k_\ell}$ and ${\bm \Upsilon}^*$. Let $\mathcal{P}_{\ell \to *}$ denote the parallel transport operator along this unique geodesic \cite{absil2008optimization}. By the triangle inequality,
\begin{equation} \label{eq:v_triangle_appE}
    \begin{aligned}
    \|{\bm v}\|_F \le & \underbrace{\| {\bm v} - \mathcal{P}_{\ell \to *}( \operatorname{grad}_{\mathcal{M}} g_r({\bm \Upsilon}^{k_\ell}) ) \|_F}_{\text{Term A}} \\
    & + \underbrace{\| \mathcal{P}_{\ell \to *} ( \operatorname{grad}_{\mathcal{M}} g_r({\bm \Upsilon}^{k_\ell}) ) \|_F}_{\text{Term B}}.
    \end{aligned}
\end{equation}
We analyze the limiting behavior as $\ell \to \infty$.

\subsubsection*{Term B}
Since the parallel transport operator $\mathcal{P}_{\ell \to *}$ is an isometry \cite{absil2008optimization}, it preserves the Frobenius norm, and thus,
\begin{equation}
    \text{Term B} 
    = \| \mathcal{P}_{\ell \to *} ( \operatorname{grad}_{\mathcal{M}} g_r({\bm \Upsilon}^{k_\ell}) ) \|_F 
    =\|\operatorname{grad}_{\mathcal{M}} g_r({\bm \Upsilon}^{k_\ell})\|_F .
\end{equation}
By the inner-loop criterion in Theorem~\ref{theorem3}, we have,
\begin{equation}
    \lim_{\ell\to\infty}\|\operatorname{grad}_{\mathcal{M}} g_r({\bm \Upsilon}^{k_\ell})\|_F = 0 
     \Longrightarrow 
    \lim_{\ell\to\infty}\text{Term B} = 0.
\end{equation}

\subsubsection*{Term A}
Expanding Term A via the linearity of parallel transport, we obtain,
\begin{equation}
    \begin{aligned}
    \text{Term A} 
    =\| (\operatorname{grad}_{\mathcal{M}} f_r({\bm \Upsilon}^*) + \sum_{i \in \mathcal{I}} \nu_i^* \operatorname{grad}_{\mathcal{M}} c_i({\bm \Upsilon}^*)) \\
    \quad - \mathcal{P}_{\ell \to *} (\operatorname{grad}_{\mathcal{M}} f_r({\bm \Upsilon}^{k_\ell}) + \sum_{i \in \mathcal{I}} \nu_i^{k_\ell} \operatorname{grad}_{\mathcal{M}} c_i({\bm \Upsilon}^{k_\ell})) \|_F .
    \end{aligned}
\end{equation}
Using the triangle inequality yields,
\begin{equation} \label{eq:termA_split_appE}
    \begin{aligned}
    &\text{Term A} \le  \underbrace{\| \operatorname{grad}_{\mathcal{M}} f_r({\bm \Upsilon}^*) - \mathcal{P}_{\ell \to *} (\operatorname{grad}_{\mathcal{M}} f_r({\bm \Upsilon}^{k_\ell})) \|_F}_{\Delta_f} \\
    & + \sum_{i \in \mathcal{I}} \underbrace{\| \nu_i^* \operatorname{grad}_{\mathcal{M}} c_i({\bm \Upsilon}^*) - \nu_i^{k_\ell} \mathcal{P}_{\ell \to *} (\operatorname{grad}_{\mathcal{M}} c_i({\bm \Upsilon}^{k_\ell}))\|_F}_{\Delta_{c,i}} .
    \end{aligned}
\end{equation}
Since $f_r(\cdot)$ is continuously differentiable, $\operatorname{grad}_{\mathcal{M}} f_r(\cdot)$ is a continuous vector field. By the continuity property of vector fields under parallel transport along minimizing geodesics \cite[Lemma A.2]{liu2020simple}, we have,
\begin{equation}
    \lim_{\ell \to \infty} \Delta_f = 0.
\end{equation}
For $\Delta_{c,i}$, adding and subtracting $\nu_i^* \mathcal{P}_{\ell \to *} (\operatorname{grad}_{\mathcal{M}} c_i({\bm \Upsilon}^{k_\ell}))$ yields,
\begin{equation}
    \begin{aligned}
    \Delta_{c,i} 
    &\le |\nu_i^*| \cdot \| \operatorname{grad}_{\mathcal{M}} c_i({\bm \Upsilon}^*) - \mathcal{P}_{\ell \to *} \operatorname{grad}_{\mathcal{M}} c_i({\bm \Upsilon}^{k_\ell}) \|_F \\
    &\quad + |\nu_i^* - \nu_i^{k_\ell}| \cdot \|\operatorname{grad}_{\mathcal{M}} c_i({\bm \Upsilon}^{k_\ell})\|_F .
    \end{aligned}
\end{equation}
The first term tends to $0$ by the continuity of $\operatorname{grad}_{\mathcal{M}} c_i(\cdot)$ and \cite[Lemma A.2]{liu2020simple}. The second term tends to $0$ since $|\nu_i^* - \nu_i^{k_\ell}| \to 0$. Therefore, $\lim_{\ell \to \infty} \Delta_{c,i} = 0$ for each $i \in \mathcal{I}$. Substituting these limits into \eqref{eq:v_triangle_appE} yields $\|{\bm v}\|_F = 0$, confirming stationarity.

Since ${\bm \Upsilon}^*$ satisfies primal feasibility, dual feasibility, complementary slackness, and stationarity, ${\bm \Upsilon}^*$ is a KKT point of problem~(\ref{overallproblem2}). \hfill $\blacksquare$

\end{appendices} 

\bibliographystyle{IEEEtran}
\bibliography{Bibliography}

\vfill

\end{document}